\documentclass[twoside]{article}

\usepackage[accepted]{aistats2024}
\usepackage[round]{natbib}

\usepackage{algorithmicx}
\usepackage{algorithm}
\usepackage{algpseudocode}
\usepackage{caption}
\usepackage{subcaption}
\usepackage{graphicx}
\usepackage{hyperref}
\usepackage{amssymb}
\usepackage{bbm}
\usepackage{amsthm}

\newtheorem{definition}{Definition}
\newtheorem{theorem}{Theorem}
\newtheorem{prop}{Proposition}
\newtheorem{lemma}{Lemma}

\newtheorem{ass}{Assumption}
\newtheorem{rem}{Remark}
\newtheorem{cor}{Corollary}

\begin{document}

%

%

\twocolumn[

\aistatstitle{Multi-Agent Learning in Contextual Games under Unknown Constraints}

\aistatsauthor{ Anna M. Maddux \And Maryam Kamgarpour }

\aistatsaddress{ EPFL Lausanne \And  EPFL Lausanne} ]

\begin{abstract}
    We consider the problem of learning to play a repeated contextual game with unknown reward and unknown constraints functions. Such games arise in applications  where each agent's action needs to belong to a feasible set, but the feasible set is a priori unknown. For example, in constrained multi-agent reinforcement learning, the constraints on the agents' policies are a function of the unknown dynamics and hence, are themselves unknown. Under kernel-based regularity assumptions on the unknown functions, we develop a no-regret, no-violation approach which exploits similarities among different reward and constraint outcomes. The no-violation property ensures that the time-averaged sum of constraint violations converges to zero as the game is repeated. We show that our algorithm, referred to as c.z.AdaNormalGP, obtains kernel-dependent regret bounds and that the cumulative constraint violations have sublinear kernel-dependent upper bounds. In addition we introduce the notion of constrained contextual coarse correlated equilibria (c.z.CCE) and show that $\epsilon$-c.z.CCEs can be approached whenever players' follow a no-regret no-violation strategy. Finally, we experimentally demonstrate the effectiveness of c.z.AdaNormalGP on an instance of multi-agent reinforcement learning.
\end{abstract}

\section{Introduction}
Several real-world problems such as auctions, traffic routing and multi-robot applications involve multiple self-interested agents that repeatedly interact with each other. Such problems can be described as \textit{repeated  games} \cite{cesa-bianchi06}. In most cases, the underlying game is unknown to the agents, consider for example, the outcome of an auction or travel times in a transportation network. Thus, a significant body of research has focused on deriving  algorithms that allow agents to learn a strategy that is aligned with their goals. 

The goal of a given learning agent in a repeated game is to maximize her cumulative reward. However, without making any assumptions on other agents' learning approach, such a goal is difficult to attain. In this setting, a non-trivial yet attainable performance measure is no-regret. A learning agent aims to learn a no-regret strategy by repeatedly playing an action and observing the corresponding reward. Existing no-regret algorithms for repeated games \cite{littlestone1994,freund95} are based on the assumption that agents can play any action. 

In many practical scenarios agents' actions are subject to a variety of constraints and a priori the set of feasible actions may not be known. These constraints can represent various requirements, such as safety regulations \cite{sun2017a, usmanova21}, fairness considerations, or task-specific limitations \cite{guo2022}. For example, in thermal control of buildings the occupants' comfort temperature and energy cost are functions of the temperature controller parameters and unknown if the dynamics of the controller are unknown.
The problem of no-regret learning under unknown constraints has been gaining increasing attention (see related work below). However, this problem has not been sufficiently explored in a game setting.

In this paper, we consider the class of \textit{contextual games with unknown constraints}. 
\textit{Repeated contextual games} \cite{sessa21} generalize the class of static games, where agents additionally observe contextual information at each round. In routing games, for instance, agents may observe the network's capacity or weather conditions.  Leveraging the context information, the agents can learn a no-regret policy, which  maps  the observed context to their actions which result in a stronger notion of regret. In constrained games, however, the agents need to operate within the boundaries of their unknown constraints. We thus address  learning a no-regret  no-violation strategy for a constrained contextual game, where no-violation implies a sublinear bound on the cumulative constraint violations.


\textbf{Related work.}
No-regret learning under unknown constraints has been extensively studied by various research fields including online convex optimization \cite{guo2022, yi2021}, black-box optimization \cite{jones22}, reinforcement learning (RL) \cite{Wei_Liu_Ying_2022}, and multi-agent reinforcement learning (MARL) \cite{chen2022}. In safety-critical applications, a line of research \cite{sui15,turchetta2019safe} has focused on developing methods, where no constraint violations are tolerated. These methods are restricted to a single player setting and require knowing an initially feasible point. For many real-world problems, however, it may be sufficient to restrict the amount of constraint violations. In designing a controller to regulate the temperature in a building \cite{dinatale2022}, for instance, deviating from the occupants' comfort temperature is tolerable but should be avoided. Several works \cite{sun2017a, Yu2020, Wei_Liu_Ying_2022, zhou22} consider the so-called soft constraints, in which the goal is to limit the sum of constraint function values over the rounds. Thus, violations at one round can be compensated in a different round. In our paper, we consider a stronger notion of constraint satisfaction over rounds. In particular, our aim is to limit the cumulative constraint \emph{violations} rather than values \cite{yi2021,guo2022,jones22}. Perhaps closest to ours is the setup of constrained online  optimization \cite{yi2021, guo2022}, where the unknown time-varying loss and constraint functions are convex. In our setup of repeated games with unknown constraints, the reward of each agent is static and the time-variations are due to the dependence of the reward function on the actions of other agents. Our goal is to leverage this structure of the game to achieve improved no-regret no-violations,  in the absence of convexity and with the inclusion of contextual information.  



Incorporating contextual information into no-regret learning has been extensively studied in the bandit literature \cite{slivkins2014a} for reward functions that depend linearly on context features \cite{li2010, chu2011, lin2022} and more generally for reward functions that are an arbitrary linear function of the contexts' images in some corresponding reproducing kernel Hilbert space (RKHS) \cite{ong2011, valko2013}. The RKHS assumption has  been considered in contextual games \cite{sessa21} and contextual no-regret algorithms have been derived for this class of games. While \cite{sessa21} leveraged the game structure to improve the regret rate compared to an online learning scenario, it did not incorporate  unknown constraints and thus, did not come with a provable violation bound.


\textbf{Contributions.} 

We propose a novel no-regret, no-violation algorithm c.z.AdaNormalGP for playing unknown contextual games with constraints. For the handling of constraints in no-regret learning, we establish a connection between playing a (contextual) game with unknown constraints and the sleeping expert problem \cite{blum1995,freund1997}. The sleeping expert problem was initially formulated to address dynamically available action sets rather than constraints on the action sets. While our feasible action set is static, the challenges in our setting is that this set is a priori unknown. 
To derive a no-regret algorithm with bounded rate of constraint violation, we exploit the assumption that similar contexts and agents' actions lead to similar reward and constraint values and similarly to \cite{sessa21,jones22} use the RKHS framework to learn the constraint functions online based on past game data. We then view an action as an asleep expert if it is infeasible with respect to the estimate of the constraint function at a given round. 


We establish a regret rate and a cumulative constraint violation rate for our algorithms c.z.AdaNormalGP. Our regret rates are the same, up to logarithmic factor in $T$ and $K$,  as those for unconstrained contextual games. Our violation rate matches that of \cite{jones22}, which only considers a static online optimization problem without any context. Concretely, c.z.AdaNormalGP with $K$ actions obtains the following bounds after $T$ rounds:
{\small
\begin{itemize}
    \item For finite context spaces, AdaNormalGP achieves {\small$\mathcal{O}(\sqrt{|\mathcal{Z}|T\log K}+\gamma_0^T\sqrt{T})$} regret and {\small$\mathcal{O}(\gamma_m^T\sqrt{T})$} cumulative constraint violation for each constraint {\small$m\in[M]$}.
    \item For infinite context spaces, AdaNormalGP achieves {\small$\mathcal{O}(T^{\frac{d+1}{d+2}}\sqrt{\log(K)}+\gamma_0^T\sqrt{T})$} regret and {\small$\mathcal{O}(\gamma_m^T\sqrt{T})$} cumulative constraint violation for each constraint {\small$m\in[M]$}.
\end{itemize}
}
We further define the new notion of
constrained contextual Coarse Correlated Equilibria (c.z.CCE) for constrained contextual games and show that
c.z.CCEs can be approached whenever agents follow a no-regret, no-violation algorithm such as c.z.AdaNormalGP.




\section{Problem Setup}\label{sec:problem_statement}
We consider a repeated contextual game among $N$ agents or players, where the two terms are used interchangeably. At every round, each player selects a feasible action and receives a payoff which depends on the action profile chosen by all players and the \textit{context} at that round. More formally, let $\mathcal{A}_i$ denote the action set of player $i$ and let $\mathcal{Z}$ represent the set of possible contexts. We assume {\small$\mathcal{A}_i^f(z):=\{a_i\in\mathcal{A}_i\ |\ g_{i,m}(a_i,z)\leq 0,\quad\forall m\in[M]\}$} to be the context-dependent feasible action set, where $z\in\mathcal{Z}$ and {\small$\{g_{i,m}:\mathcal{A}_i\times\mathcal{Z}\rightarrow\mathbb{R}\}_{m\in[M]}$} is a set containing $M$ constraint functions. Furthermore, define {\small$r_i:\mathcal{A}\times\mathcal{Z}\rightarrow[0,1]$} to be the reward function of each player $i$, where {\small$\mathcal{A}:=\mathcal{A}_1\times\ldots\times\mathcal{A}_N$} is the joint action space. We assume that each agent's reward function $r_i$ and constraint functions $g_{i,m}$, {\small$m\in[M]$}, are unknown to her, and thus also her context-dependent feasible action set {\small$\mathcal{A}_i^f(z)$ for all $z\in\mathcal{Z}$}. Then, a repeated contextual game with unknown constraints proceeds as follows. At every round $t$, context $z^t$ is revealed. The players observe $z^t$ and, based on it, simultaneously each player $i\in[N]$ selects an action $a_i^t\in\mathcal{A}_i$ which needs to be feasible, i.e., {\small$a_i^t\in\mathcal{A}_i^f(z^t)$}. Then, players obtain rewards $r_i(a_i^t,a_{-i}^t,z^t)$ and constraints $g_{i,m}(a_i^t,z^t)$, where {\small $m\in[M]$} and {\small$a_{-i}^t:=(a_1^t,\ldots,a_{i-1}^t,a_{i+1}^t,\ldots,a_N^t)$}.

We define $\Pi_i$ to be the set of all policies $\pi_i:\mathcal{Z}\rightarrow\mathcal{A}_i$ mapping contexts to actions. After $T$ rounds, the performance of player $i$ is measured by the \textit{constrained contextual regret}:
{\small
\begin{align}\label{def:contextual_regret}
\begin{split}
    R_i^T=&\max_{\substack{\pi_i\in\Pi_i}} \sum_{t=1}^T r_i(\pi_i(z^t),a_{-i}^t,z^t)-\sum_{t=1}^T r_i(a_i^t,a_{-i}^t,z^t)\\
    &\text{s.t.}\quad g_{i,m}(\pi_i(z^t),z^t)\leq 0,\quad\forall m\in[M], \forall t\geq 1.
\end{split}
\end{align}
}

We define the cumulative constraint violation for each constraint $g_{i,m}$ with $m\in[M]$:
{\small
\begin{align}\label{def:constraint_violation}
    \mathcal{V}_{i,m}^T=\sum_{t=1}^T[g_{i,m}(a_i^t,z^t)]_+,
\end{align}
}

where $[x]_+:=\max\{0,x\}$. The constrained contextual regret, short regret, measures the gain player $i$ could have achieved by following her best feasible fixed policy in hindsight had the sequence of other players' actions {\small$\{a_{-i}\}_{t=1}^T$}, the reward $r_i(\cdot)$ and the constraint $g_{i,m}(\cdot)$ functions been known to her. The cumulative constraint violation measures the accumulated violations over all rounds $T$ for each constraint function. A strategy is \textit{no-regret} for player $i$ if {\small$R_i^T/T\rightarrow 0$} as {\small$T\rightarrow\infty$}. Similarly, we call a strategy \textit{no-violation} for player $i$ if {\small$\mathcal{V}_{i,m}^T/T\rightarrow 0$} for all $m\in[M]$ as {\small$T\rightarrow\infty$}. 

The problem we address in this paper is the design  of an algorithm that is simultaneously \textit{no-regret} and \textit{no-violation}. To ensure that the learning problem is meaningful, we assume that a feasible policy exists for each player. 

\begin{ass}[Feasibility assumption]\label{ass:feasibility_context}
    We assume that for each agent $i\in\mathcal{N}$ the following problem
    {\small
    \begin{align}\label{problem_context}
    \begin{split}
        &\max_{\pi_i\in\Pi_i}\quad \sum_{t=1}^T r_i(\pi_i(z^t),a_{-i}^t,z^t)\\
        &\text{s.t.}\quad g_{i,m}(\pi_i(z^t),z^t)\leq 0,\quad\forall m\in[M], \forall t\geq 1
    \end{split}
    \end{align}
    }
    
    is feasible with optimal constrained solution $\overline{\pi}_i^*$.
\end{ass}
Observe that this is a natural assumption for requiring no-cumulative violation. In particular, in a non-contextual game setting this assumption is equivalent existence of a feasible action.

\subsection{Feedback model and regularity assumption}



At the end of each round, player $i$ receives feedback information that she can use to update her strategy and thereby improve her performance in terms of regret and constraint violations. In a full-information feedback model player $i$ observes her reward for any action $a_i\in\mathcal{A}_i$ given the other players' actions $a_{-i}^t$ and the observed context $z^t$, i.e., $\mathbf{r}_i^t=[r_i^t(a_i,a_{-i}^t,z^t)]_{a_i\in\mathcal{A}_i}$. A more realistic feedback model is bandit feedback, where player $i$ observes her reward $r_i^t(a_i^t,a_{-i}^t,z^t)$ for the played action profile $a^t$. In the constrained setting, we extend the term bandit feedback such that after each round player $i$ additionally observes her constraint values $g_{i,m}(a_i^t,z^t)$, for all $m\in[M]$, for her played action $a_i^t$. Additionally, in many practical settings player $i$ can observe the other players' played actions $a_{-i}^t$ or some aggregative function $\gamma(a_{-i}^t)$ thereof, e.g. in routing games, agents may observe the total occupancy of each edge in the road network or in electricity markets, agents may observe the aggregate load.

\begin{ass}[Feedback assumption]\label{ass:feedback_context}
    We consider a noisy bandit feedback model, where, at every round $t$, player $i$ observes a noisy measurement of the reward $\Tilde{r}_i^t=r_i(a_i^t,a_{-i}^t,z^t)+\epsilon_i^t$ and a noisy measurement of each constraint $\Tilde{g}_{i,m}^t=g_{i,m}(a_i^t)+\epsilon_{i,m}^t$. The noise $\epsilon_i^t$ is $\sigma_{i,0}$-sub-Gaussian and $\epsilon_{i,m}^t$ is $\sigma_{i,m}$-sub-Gaussian for each $m\in[M]$. Furthermore, we assume player $i$ also observes the played actions $a_{-i}^t$ of the other players.
\end{ass}

Repeated contextual games with unknown constraints generalize the class of repeated static games in two ways: Firstly, as discussed in \cite{sessa21, valko2013}, the game may change from round to round due to a potentially different context $z^t$. Secondly and so far unaddressed by the literature, the action set $\mathcal{A}_i$ of each player may be subject to unknown (context-dependent) constraints. This makes this generalization challenging as a priori a player does not know whether an action is feasible or not.

Attaining \textit{no-regret} and \textit{no-violation} is impossible  without any regularity assumptions on the reward and constraint functions \cite{srinivas2009gaussian}. In many practical settings, including routing games and thermal control, similar inputs lead to similar rewards and constraint values which implies some regularity on the reward and constraint functions. In the application of thermal control, for instance, similar choices of control input parameters and set point temperatures result in a similar energy consumption.  

\begin{ass}[Regularity assumption]\label{ass:regularity_context}
    We assume the action set $\mathcal{A}_i$ is finite with {\small$|\mathcal{A}_i|=K$} and consider a general context set $\mathcal{Z}\subseteq\mathbb{R}^d$. Let {\small$\mathcal{D}=\mathcal{A}\times\mathcal{Z}$}. We assume the unknown reward function {\small$r_i:\mathcal{A}\times\mathcal{Z}\rightarrow [0,1]$} has a bounded norm {\small$\|r_i\|_{k_{i,0}}=\sqrt{\langle r_i,r_i\rangle_{k_{i,0}}}\leq B_{i,0}$} in a reproducing kernel Hilbert space (RKHS, see \cite{rasmussen2005}) associated with a positive semi-definite kernel function $k_{i,0}(\cdot,\cdot)$. The RKHS is denoted by {\small$\mathcal{H}_{k_{i,0}}(\mathcal{D})$}. We further assume bounded variance by restricting $k_{i,0}(x, x')\leq 1$ for all $x,x'\in\mathcal{D}$. Let {\small$\mathcal{D}_i=\mathcal{A}_i\times\mathcal{Z}$}. Similarly, we assume each unknown constraint function {\small$g_{i,m}:\mathcal{A}_i\times\mathcal{Z}\rightarrow[0,1]$}, {\small$m\in[M]$}, has a bounded norm {\small$\|g_{i,m}\|_{k_{i,m}}\leq B_{i,m}$} in a RKHS associated with a positive semi-definite kernel function $k_{i,m}(\cdot,\cdot)$. The RKHS is denoted by {\small$\mathcal{H}_{k_{i,m}}(\mathcal{D}_i)$} and we assume that $k_{i,m}(x_i, x_i')\leq 1$ for all $x_i,x_i'\in\mathcal{D}_i$. 
\end{ass}

For general RKHS {\small$\mathcal{H}_k(\mathcal{X})$} with positive semi-definite kernel function $k(\cdot,\cdot)$, the RKHS norm $\|f\|_k$  measures smoothness\footnote{This can be seen from {\small$|f(x)-f(x')|=|\langle f,k(x,\cdot)-k(x',\cdot)\rangle|\leq \|f\|_k\|k(x,\cdot)-k(x',\cdot)\|_k$} for any {\small$f\in\mathcal{H}_k(\mathcal{X})$}, using the reproducing property and Cauchy-Schwarz inequality.} of {\small$f\in\mathcal{H}_k(\mathcal{X})$} with respect to the kernel function $k(\cdot,\cdot)$, while the kernel $k(\cdot,\cdot)$ encodes similarities between different points $x,x'\in\mathcal{X}$. 
Assuming that some unknown function has bounded norm in a RKHS is standard in black-box optimization \cite{srinivas2009gaussian,chowdhury2017kernelized} and  was recently also exploited in repeated (contextual) games \cite{sessa2019no,sessa21}. It is analogous to assuming bounded weights in linear parametric bandits \cite{chu2011,abbasi-yadkori2011}. The bound of $1$ on $k(x,x')$ is only required to obtain scale-free bounds. Otherwise, if $k(x, x') \leq n$ for all $x,x'\in\mathcal{X}$ or $\|f\|_k \leq nB$, our bounds increase by a factor of $n$ \cite{agrawal2014}. Typical kernel choices are the \textit{Squared Exponential}, \textit{Matérn} and the \textit{polynomial} kernel:
{\small
\begin{align*}
    &k_{SE}(x,x') = \exp\bigg(-\frac{s^2}{2l^2}\bigg),\\
    & k_{\text{Matérn}}(x,x') = \frac{2^{1-\nu}}{\Gamma(\nu)}\bigg(\frac{s\sqrt{2\nu}}{l}\bigg)^\nu B_\nu\bigg(\frac{s\sqrt{2\nu}}{l}\bigg)^\nu,\\
    & k_{\text{poly}}(x,x')=\bigg(b+\frac{x^Tx'}{l}\bigg)^d,
\end{align*}
}

where $l,\nu, b>0,$ are kernel hyperparameters, $s=\|x-x'\|_2$ encodes similarities between points $x,x'\in\mathcal{X}$, and $B_\nu(\cdot)$ is the modified Bessel function. We emphasize that Assumption \ref{ass:feedback_context} and \ref{ass:regularity_context} allow player $i$ to learn about her unknown reward $r_i$ and constraint $g_{i,m}$ functions.

The Gaussian process (GP) framework can be used to learn unknown functions {\small$f\in\mathcal{H}_k(\mathcal{X})$} with {\small$\|f\|_k<\infty$} since such functions can be modelled as a sample from a GP \cite[Section~6.2]{rasmussen2005}. A GP over $\mathcal{X}$ is a probability distribution over functions {\small$f(x) \sim \mathcal{GP}(\mu(x), k(x, x'))$}, specified by its mean and covariance functions $\mu(\cdot)$ and $k(\cdot, \cdot)$, respectively. A GP prior {\small$\mathcal{GP}(0,k(\cdot,\cdot))$} over the initial distribution of the unknown function $f$, where $k(\cdot,\cdot)$ is the kernel function associated with the RKHS {\small$\mathcal{H}_k(\mathcal{X})$}, is used to capture the uncertainty over $f$. Then, the GP framework can be used to predict function values $f(x)$ for any points $x\in\mathcal{X}$ based on a history of measurements $\{y^{\tau}\}_{\tau=1}^t$ at points $\{x^{\tau}\}_{\tau=1}^t$ with {\small$y^\tau = f(x^\tau)+\epsilon^\tau$} and {\small$\epsilon^\tau\sim\mathcal{N}(0,\sigma^2)$}. Conditioned on the history of measurements, the posterior distribution over $f$ is a GP with mean and variance functions:
{\small
\begin{align}
    &\mu^t(x) = \textbf{k}^t(x)^\top (\textbf{K}^t+\sigma^2 \textbf{I}^t)^{-1}\textbf{y}^t \label{eq:mu}\\   
    &(\sigma^t)^2(x)=k(x,x)-\textbf{k}^t(x)^\top(\textbf{K}^t+\sigma^2 \textbf{I})^{-1}\textbf{k}^t(x), \label{eq:var}
\end{align}
}

where $\textbf{k}^t(x)=[k(x^\tau,x)]_{\tau=1}^t$, $\textbf{y}^t=[y^\tau]_{\tau=1}^t$, and $\textbf{K}^t=[k(x^\tau,x^{\tau'})]_{\tau,\tau'=1}^t$ is the kernel matrix.


In the next section, we present our algorithm which exploits the feedback model and the regularity assumption to derive a no-regret no-violation algorithm. We focus on the perspective of a single player and drop the subscript $i$ wherever this is possible without causing confusion.

\section{The c.z.AdaNormalGP Algorithm}\label{sec:c.z.AdaNormalGP}

\begin{algorithm}[!t]
\caption{c.z.AdaNormalGP algorithm}\label{alg:contextual_constrainedGPHedge}
\begin{algorithmic}[1]
    \For{\texttt{$t=1,\ldots,T$}}
        \State \texttt{Observe context $z^t$.}
        \If{{\tiny$\exists m\in[M]$} such that {\tiny$\min_{\pi_i\in\Pi_i}\text{LCB}_m^t(\pi_i(z^t),z^t)>0$}}
            \State \texttt{Declare infeasibility.}
        \EndIf
        \State \texttt{Compute distribution $p^t(z^t)$ using $z^t$ and\\
        \hspace{\algorithmicindent}{\small$[a_i^\tau,a_{-i}^\tau,z^\tau,\Tilde{r}^\tau]_{\tau=1}^{t-1}$}}
        \State \texttt{Sample $a_i^t\sim \overline{p}^t(z^t)$, where}
        \begin{align*}
            {\tiny\overline{p}_{a_i}^t(z^t)=\begin{cases}
        \frac{p_{a_i}^t(z^t)}{\sum_{k\in\mathcal{A}_i} \overline{p}_{k}^t(z^t)}, & \text{if $\text{LCB}_m^t(a_i,z^t)\leq 0$ $\forall m\in[M],$}\\
        0, & \text{otherwise}.
        \end{cases}}
        \end{align*}
        \State \texttt{Observe noisy reward $\Tilde{r}^t$ and constraints\\
        \hspace{\algorithmicindent} $\Tilde{g}_{m}^t$ and action profile $a_{-i}^t$.}
        \State \texttt{Append {\small$(a_i^t,a_{-i}^t,\Tilde{r}^t,z^t)$} and {\small$(a_i^t,\Tilde{g}_{m}^t,z^t)$} to the\\
        \hspace{\algorithmicindent}history of play.}
        \State \texttt{Update $\mu_0^t, \sigma_0^t$ and $\mu_m^t, \sigma_m^t$ via \eqref{eq:mu}-\eqref{eq:var}.}
    \EndFor
\end{algorithmic}
\end{algorithm}

In our approach, we build upper confidence bounds on the reward function and lower confidence bounds on the constraint functions using the GP framework proposed in the previous section. We thereby obtain optimistic estimates of the reward and the constraints which enable us to derive bounds for the regret and cumulative constraint violations. Particularly, for the constraints, we use lower confidence bounds to ensure that even in early rounds when we have poor estimates of the constraints, we can find feasible actions. Considering optimistic estimates of the constraints is not limiting since as the amount of past game data increases over the rounds uncertainty about the unknown constraints decreases and thus the confidence bounds shrink towards the true constraint functions

Making use of past game data {\small$[(a_i^\tau,a_{-i}^\tau,\Tilde{r}^\tau,\Tilde{g}_m^\tau)]_{\tau=1}^{t-1}$}, at round $t$,  the posterior mean $\mu_m^t(\cdot)$ and posterior variance $(\sigma_m^t)^2(\cdot)$, for {\small$m\in\{0\}\cup[M]$}, can be used to obtain an upper and a lower confidence bound on $r$ and $g_m$, respectively:
{\small
\begin{align}
    &\text{UCB}_m^t(x) := \mu_m^{t-1}(x) + \beta_m^t\sigma_m^{t-1}(x),\quad\forall x\in\mathcal{X}\label{eq:ucb}\\
    &\text{LCB}_m^t(x) := \mu_m^{t-1}(x) - \beta_m^t\sigma_m^{t-1}(x),\quad\forall x\in\mathcal{X}\label{eq:lcb},
\end{align}
}

where {\small$\mathcal{X}=\mathcal{A}\times\mathcal{Z}$} for the reward function $r$ and {\small$\mathcal{X}=\mathcal{A}_i\times\mathcal{Z}$} for each constraint function $g_m$. 
Parameter $\beta_m^t$ controls the width of the confidence bound. Importantly, if $\beta_m^t$ is specified adequately, the reward $r$ and the constraint $g_{m}$ functions are bounded from below and above by the lower {\small$\text{LCB}_m^t$} and the upper {\small$\text{UCB}_m^t$} confidence bound with high probability. For the reward function, for example, if $\beta_0^t$ is specified as {\small$\beta_0^t=B_0+\sigma_0\sqrt{2(\gamma_0^{t-1}+1+\log(1/\delta)}$}, then {\small$\text{LCB}_0^t(x)\leq r(x)\leq\text{UCB}_0^t(x)$}, for all $x\in\mathcal{X}$ and $t \geq 1$, with probability $1-\delta$ \cite{chowdhury2017kernelized}. 
Here, $\gamma_0^{t}$, the \textit{maximum information gain} \cite{srinivas2009gaussian}, is a kernel-dependent quantity which for the reward $r$ is defined as:
{\small
\begin{align*}
    \gamma_0^t:= \max_{A\subset\mathcal{X}:|A|=t} I(y_A;r_A).
\end{align*}
}

Quantity $I(y_A,r_A)$ denotes the mutual information between the true reward values $\textbf{r}_A=[r(x)]_{x\in A}$ and the noisy reward measurements $\textbf{y}_A=[r(x)+\epsilon]_{x\in A}$.
It quantifies the maximal reduction in uncertainty about $r$ after observing points $A\subset\mathcal{X}$ and their corresponding noisy measurements $\textbf{y}_A$. The maximum information gain $\gamma_m^t$ of each constraint function $g_{m}$ is defined analogously.


We now introduce constrained contextual AdaNormalGP (c.z.AdaNormalGP) algorithm for repeated contextual games with a priori unknown constraints on the action set of each player. At each round, c.z.AdaNormalGP observes context $z^t$ and checks whether a feasible action exists with respect to any $\text{LCB}_m^t$ estimate of constraint $g_m$, since the constraints $g_{m}$ are unknown. The algorithm then computes a context-dependent distribution\footnote{The distribution update rule is inspired by the update rule of AdaNormalHedge in \cite{luo15}.} $p^t(z^t)\in\Delta_K$ over the action set $\mathcal{A}_i$ based on the $\text{UCB}_0^t$ estimate of the reward $r$ and the observed context $z^t$. Here $\Delta_K$ denotes the $K$-dimensional probability simplex. The distribution is renormalized to assign zero weight to an action $a_i\in\mathcal{A}_i$ which is infeasible with respect to any $\text{LCB}_m^t$ estimate, $m\in[M]$. A feasible action $a_i^t$ is sampled from the renormalized distribution. Once the computation rule of $p^t(z^t)$ (line $6$ of Algorithm \ref{alg:contextual_constrainedGPHedge}) is specified, c.z.AdaNormalGP is well-defined. 

Consider a fixed context $z\in\mathcal{Z}$. We establish that playing a repeated game with unknown constraints is equivalent to playing a sleeping expert problem \cite{blum1995,freund1997}. In the sleeping expert problem, at each round each expert (action $a_i$) reports whether it is awake (feasible) or asleep (infeasible). Since the constraint functions $g_m$ are unknown, c.z.AdaNormalGP maintains an optimistic estimate of the set of feasible actions at each round, defined as the set of actions $\mathcal{A}_i$ for which, for all $m\in[M]$, $\text{LCB}_m^t(a_i,z^t)\leq 0$. Thus, in a repeated game with unknown constraints the set of awake experts at round $t$ corresponds to the set of feasible actions with respect to the $\text{LCB}_m^t$ estimate of each constraint $g_m$ for all $m\in[M]$. Leveraging this connection, we specify a computation rule of $p^t(z^t)$ based on the update rule of the sleeping expert algorithm AdaNormalHedge \cite{luo15}. In the following, we specify $p^t(z^t)$ when the context space $\mathcal{Z}$ is finite (Algorithm \ref{alg:strategy_finite_Z}). In Appendix \ref{app:general_regret_bounds_infinite_Z} we specify a computation rule $p^t(z^t)$ for an infinite context space and provide bounds on the regret and the cumulative constraint violations for this setting. (Algorithm \ref{alg:strategy_infinite_Z}).


\begin{rem}\label{rem:general_bounds}

The above established connection implies that, for a fixed context $z\in\mathcal{Z}$, any update rule that corresponds to an update rule of a sleeping expert problem can be used to obtain high probability regret bounds. Moreover, since any sleeping expert problem can be reduced to a regular expert problem \cite{freund1997}, any update rule that corresponds to the update rule of the reduced expert problem can be used to obtain high probability regret bounds for repeated games with unknown constraints. These bounds depend on the regret which the expert algorithm obtains. We formally introduce and study general regret bounds in Appendix \ref{app:general_regret_bounds} for the case of finite and infinite context spaces.
\end{rem}


\subsection{Finite number of contexts}\label{sec:finite_context}

When the context set $\mathcal{Z}$ is finite, on a high level the distribution $p^t(z^t)\in\Delta_K$ (line 6 of Algorithm \ref{alg:contextual_constrainedGPHedge}) is computed by maintaining a distribution for each context $z\in\mathcal{Z}$ which is updated whenever $z$ is observed. Thus, playing a repeated contextual game reduces to playing $|\mathcal{Z}|$ repeated static games. For a specific context $z\in\mathcal{Z}$ the update rule is defined as the general update rule of AdaNormalHedge \cite{luo15}. Under this computation rule of $p^t(z^t)$, summarized in Algorithm \ref{alg:strategy_finite_Z}, c.z.AdaNormalGP achieves the following high-probability bounds.

\begin{theorem}\label{thm:finite_Z}
    Fix $\delta\in(0,1)$. Under Assumptions \ref{ass:feasibility_context}-\ref{ass:regularity_context}, if a player plays according to c.z.AdaNormalGP with $p^t(z^t)$ computed according to Algorithm \ref{alg:strategy_finite_Z} and {\small$\beta_m^t=B_m+\sigma_m\sqrt{2(\gamma_m^{t-1}+1+\log(2(M+1)/\delta))}$} for all {\small$m\in\{0\}\cup[M]$}, then with probability at least $1-\delta$:
    {\small 
    \begin{align}
    &R^T=\mathcal{O}\big(\sqrt{|\mathcal{Z}|T(\log(K)+\log(B)+\log(1+\log(K))}+\nonumber\\
&\quad\quad\quad\quad\sqrt{T\log(2/\delta)}+\beta_0^T\sqrt{T\gamma_0^T}\big) \label{eq:thm1_regret}\\
    &\mathcal{V}_{m}^T=\mathcal{O}\big(\beta_m^T\sqrt{T\gamma_m^T}\big),\hspace{0,3cm}\forall m\in[M],
    \end{align}
    }
    where {\small$B=1+\frac{3}{2}\frac{1}{K}\sum_{a_i=1}^K (1+\log(1+C_i^t(a_i)))\leq \frac{5}{2}+\frac{3}{2}\log(1+T)$} \label{eq_thm1_constraint}.
\end{theorem}

We provide a detailed proof in Appendix \ref{app:finite_Z}. 

\begin{cor}
Under the same assumptions as in Theorem \ref{thm:finite_Z}, if a player plays according to c.z.AdaNormalGP with $|\mathcal{Z}|=1$, then with probability $1-\delta$ the following regret\footnote{The static games setup is recovered by assuming $z^t = z^0$ for all $t$ and the regret is defined as {\tiny$R^T=\max_{\substack{a_i\in\mathcal{A}_i}} \sum_{t=1}^T r_i(a_i,a_{-i}^t)-\sum_{t=1}^T r_i(a_i^t,a_{-i}^t)
    $} subject to {\tiny$g_{i,m}(a_i,z^t)\leq 0,\quad \forall m\in[M],\forall t\geq 1$}.} bound for repeated static games with unknown constraints holds:
    {\small
    \begin{align*}
    R^T=\mathcal{O}\big(&\sqrt{T(\log(K)+\log(B)+\log(1+\log(K))}+\\
    &\sqrt{T\log(2/\delta)}+\beta_0^T\sqrt{T\gamma_0^T}\big).
    \end{align*}
    }
The high probability bound on the cumulative constraint violations is the same as in Theorem \ref{thm:finite_Z}.
\end{cor}

In repeated contextual games, c.GPMW \cite{sessa21} algorithm  achieves {\small$\mathcal{O}(\sqrt{|\mathcal{Z}|T\log K}+\sqrt{T\log(2/\delta)}+\beta_0^T\sqrt{T\gamma_0^T})$} which is the same as our regret bound up to logarithmic terms. Following the same arguments as in the proof of Theorem \ref{thm:finite_Z}, AdaNormalHedge can be extended to the contextual setting which implies a {\small$\mathcal{O}(\sqrt{|\mathcal{Z}|T(\log( K)+\log(B)+\log(1+\log K))}+\sqrt{T\log(2/\delta)})$} bound on the contextual sleeping expert regret. Note that, unlike AdaNormalHedge, c.z.AdaNormalGP does not consider full-information, but obtains a regret bound that is equal up to the term {\small$\mathcal{O}(\beta_0^T\sqrt{T\gamma_0^T})$}. Furthermore, both c.GPMW and AdaNormalHedge do not provide bounds on the cumulative constraint violations. In constrained black-box optimization, CONFIG \cite{jones22} algorithm obtains a bound on the cumulative constraint violations which is the same as ours.

\begin{algorithm*}[t]
\caption{Update rule for finite context space $\mathcal{Z}$}\label{alg:strategy_finite_Z}
\hspace*{\algorithmicindent} \textbf{Set} $R^1_{a_i}(z)=0$, $C^1_{a_i}(z)=0$ and let $p^1(z)$ be the uniform distribution $\forall z\in\mathcal{Z}$.
\begin{algorithmic}[1]
    \For {\texttt{$t=2,\ldots,T$}}
        \State \texttt{Compute reward estimate $\hat{r}^t$ with {\small$\hat{r}^t(a_i)=\min\{1,\text{UCB}_0^t(a_i,a_{-i}^t,z^t)\},\quad \forall a_i\in\mathcal{A}_i$}.}
        \State \texttt{Set for $z^t$ and every $a_i\in\mathcal{A}_i$ with {\small$\mathbb{E}[\hat{r}^t(a_i^t)]=\overline{p}^t(z^t)^\top \hat{r}^t$}: 
        {\small\begin{align}
        &R^t_{a_i}(z^t)=\begin{cases}
        R^{t-1}_{a_i}(z^t)+(\hat{r}^t(a_i)-\mathbb{E}[\hat{r}^t(a_i^t)]), & \text{if $\text{LCB}_m^t(a_i,z^t)\leq 0$, $\forall m\in[M],$}\\
        R_{a_i}^{t-1}(z^t), & \text{otherwise}.
        \end{cases} \label{alg:strategy_finite_Z_R}\\
        & C_{a_i}^t(z^t)=\begin{cases}
        C_{a_i}^{t-1}(z^t)+|\hat{r}^t(a_i)-\mathbb{E}[\hat{r}^t(a_i^t)]|, & \text{if $\text{LCB}_m^t(a_i,z^t)\leq 0$, $\forall m\in[M],$}\\
        C_{a_i}^{t-1}(z^t), & \text{otherwise}.
        \end{cases} \label{alg:strategy_finite_Z_C}
        \end{align}}
        }
        \State \texttt{Update   {\small $\label{alg:strategy_finite_Z_update}
            p_{a_i}^{t+1}(z^t)=\frac{w(R_{a_i}^t(z^t),C_{a_i}^t(z^t))}{\sum_{k\in\mathcal{A}_i} w(R_k^t(z^t),C_k^t(z^t))}$,} where {\small$w(R,C)=\frac{1}{2}(\exp([R+1]_+^2/3(C+1))-\exp([R-1]_+^2/3(C+1)))$}.}
    \EndFor
\end{algorithmic}
\end{algorithm*}


\subsection{Game equilibria}

In the previous section, we focused on the perspective of a single player and did not make any assumptions on how the other players select their actions. It is known that if all players follow a no-regret dynamics a coarse correlated equilibrium (CCE) can be approached \cite{hart2000simple}. A CCE of a game is a probability distribution $\rho\in\Delta^{|\mathcal{A}|}$ such that
{\small
\begin{align*}
    \mathbb{E}_{a\sim\rho} \big[r_i(a)\big]\geq \mathbb{E}_{a\sim\rho} \big[r_i(a_i',a_{-i})\big], \quad\forall i\in\mathcal{N}, a_i'\in\mathcal{A}_i.
\end{align*}
}

In other words, if an action profile is sampled from a CCE, then in expectation each player is better off following the sampled action than playing any other action. To deal with repeated contextual games \cite{sessa21} introduces the notion of contextual CCEs (z.CCE) and extends the results of \cite{hart2000simple} to this setting. In the following, we introduce the notion of constrained contextual CCEs (c.z.CCE) to address the fact that only a subset of a player's action set is feasible. 

\begin{definition}\label{def:czCCE}
    Let $z^1,\ldots,z^T$ be the revealed sequence of contexts. A constrained contextual coarse correlated equlibrium is a joint policy $\rho:\mathcal{Z}\rightarrow\Delta^{|\mathcal{A}|}$ such that
    {\small
    \begin{align*}
        &\frac{1}{T}\sum_{t=1}^T\mathbb{E}_{a\sim\rho_{z^t}}\big[[g_{i,m}(a_i, z^t)]_+\big]\leq0\\
        & \frac{1}{T} \sum_{t=1}^T\mathbb{E}_{a\sim\rho_{z^t}} \big[r_i(a,z^t)\big]\geq \frac{1}{T} \sum_{t=1}^T\mathbb{E}_{a\sim\rho_{z^t}} \big[r_i(\pi_i(z^t),a_{-i},z^t)\big],
    \end{align*}
    }
    
    for all $i\in\mathcal{N}$, $m\in[M]$ and for all feasible $\pi_i\in\Pi_i$, i.e., $g_{i,m}(\pi_i(z^t),z^t)\leq 0$ for all $t=1,\ldots,T$.
    An $\epsilon$-c.z.CCE is a joint policy $\rho:\mathcal{Z}\rightarrow\Delta^{|\mathcal{A}|}$ such that the above two inequalities are satisfied up to $\epsilon>0$ accuracy.
\end{definition}
The above definition requires that for a fixed context an action profile sampled from a c.z.CCE is feasible in expectation. To illustrate, suppose $\rho$ is a c.z.CCE and suppose there is a trusted device that, for any context $z^t$, samples a joint action from $\rho(z^t)$. When complying with such a device in expectation each player's action is feasible and each player is better off than when using any other feasible policy $\pi_i$. 

Next we show that a c.z.CCE can be approached when all players minimize their constrained contextual regret. We first define the empirical joint policy $\rho^T$ as round $T$ as 
{\small
\begin{align}\label{eq:emirical_joint_policy}
    \rho_z^T(a) = \begin{cases}
        \frac{1}{T_z}|\{t\in[T]: z^t=z, a^t=a\}|, &\text{if $z\in\mathcal{Z}^T$}\\
        \frac{1}{|\mathcal{A}_i|}, &\text{otherwise},
    \end{cases}
\end{align}
}

where $T_z=|\{t\in[T] : z^t=z\}|$ is the number of rounds that context $z$ was revealed and $\mathcal{Z}^T\subseteq\mathcal{Z}$ is the set of revealed contexts. Note that for unobserved contexts $\rho(z)$ can be any arbitrary  distribution.

\begin{prop}\label{prop:c.z.CEE}
    After $T$ rounds, the empirical joint policy $\rho^T$ is an $\epsilon$-c.z.CCE of the played constrained contextual game with $\epsilon\leq \max\{\epsilon_1,\epsilon_2\}$, where $\epsilon_1=\max_{i\in\mathcal{N}} R_i^T/T$ and $\epsilon_2=\max_{i\in\mathcal{N}, m\in[M]} \mathcal{V}_{i,m}^T/T$.
\end{prop}

We provide a detailed proof in Appendix \ref{app:game_equilibria}. Proposition \ref{prop:c.z.CEE} implies that if all players follow no-regret, no-violation dynamics then the empirical join policy $\rho^T$ converges to a $c.z.CCE$ of the constrained contextual game, as $T\rightarrow\infty$.

\section{Experiments}\label{sec:experiments}
Here, we focus on the application of our approach to a  realistic multi-building temperature control design problem. In Appendix \ref{app:random_game} we benchmark our algorithm against existing approaches for randomly generated contextual $N$-player game. 

\subsection{Temperature controller design}\label{sec:experiments_robots}
Optimizing a building's energy consumption  is of prime importance, both for ecological and economical reasons, as nearly $40\%$ of the total energy is consumed by the building sector \cite{laustsen2008energy} and the price of electricity has been surging over the last few years. Several recent works have looked at this problem from a single agent perspective, using Bayesian optimization \cite{fiducioso2019safe}, control or reinforcement learning \cite{dinatale2022}. These approaches do not account for the dependence of the electricity price on the aggregate consumption of energy of multiple buildings. Our work particularly accounts for this dependence, generalizing the single player approaches to a multi-player game setting. 

We tune temperature controllers for multiple buildings using c.AdaNormalGP\footnote{We do not consider contexts in this application and thus refer to Algorithm \ref{alg:contextual_constrainedGPHedge} as c.AdaNormalGP.} with the goal of reducing the buildings' energy cost while satisfying operational constraints, such as pleasant room temperature. 

We abstract this problem as follows: We consider three buildings {\small$i\in\{1,2,3\}$} with a single thermal zone which we simulate using \textbf{Energym} \cite{scharnhorst2021energym} and \textbf{Modelica} \cite{fritzson1998modelica}. Each building is equipped with a heatpump which, similarly to \cite{jones22}, is regulated by a proportional controller \cite{ang2005pid}. The decision variable of each building corresponds to the parameters {\small$a_i=(K_i, \text{SP}_i, \text{ST}_i)$} of the building's proportional controller, where $K_i$ is the controller gain, $\text{SP}_i$ is the set point temperature during the day, and $\text{SP}_i$ the switching time from day- to nighttime set point temperature $\text{ST}_i$. Fixing the selected control parameters $a_i$, we simulate the energy consumption and temperature of the building based on the Energym and Modelica models for a fixed period of time. After this period, building observes its electricity cost $J_i(a_i,a_{-i})$ which is an aggregate function of the total energy consumption, and observes the maximum temperature deviation $g_i(a_i)$ from the setpoint temperature. Based on this feedback each building aims to tune its control parameters to reduce its energy cost while at the same time ensuring that the temperature stays within some occupant-specific comfort range. 

In our experiments, we discretize space of control parameters and let each building select its control parameters $a_i$ according to c.AdaNormalGP. We simulate the energy consumption and temperature for a period of $2$ days. Similarly to \cite{hall2022receding}, we model the cost function as an affine function of the aggregate consumption:
{\small
\begin{align*}
    J_i(a_i,a_{-i})=\sum_{h=1}^{48} \bigg(c_1\sum_{i=1}^3 l_i^h(a_i) + c_2\bigg)l_i^h(a_i),
\end{align*}
}

where $l_i^h(a_i)$ denotes the energy consumption of building $i$ at hour $h$. The positive constants $c_1$ and $c_2$ represent different price rates which depend on whether the energy is consumed during peak or off-peak hours. We further require that the maximum temperature deviation from the setpoint temperature during the night remains below an occupant-specific threshold: 
{\small
\begin{align*}
    \max_{h\in 1,\ldots,48} g_i^h(a_i) - g_{i,\text{thr}}\leq 0.
\end{align*}
}


\begin{figure}[htbp]
    \centering
    \includegraphics[scale = 0.5]{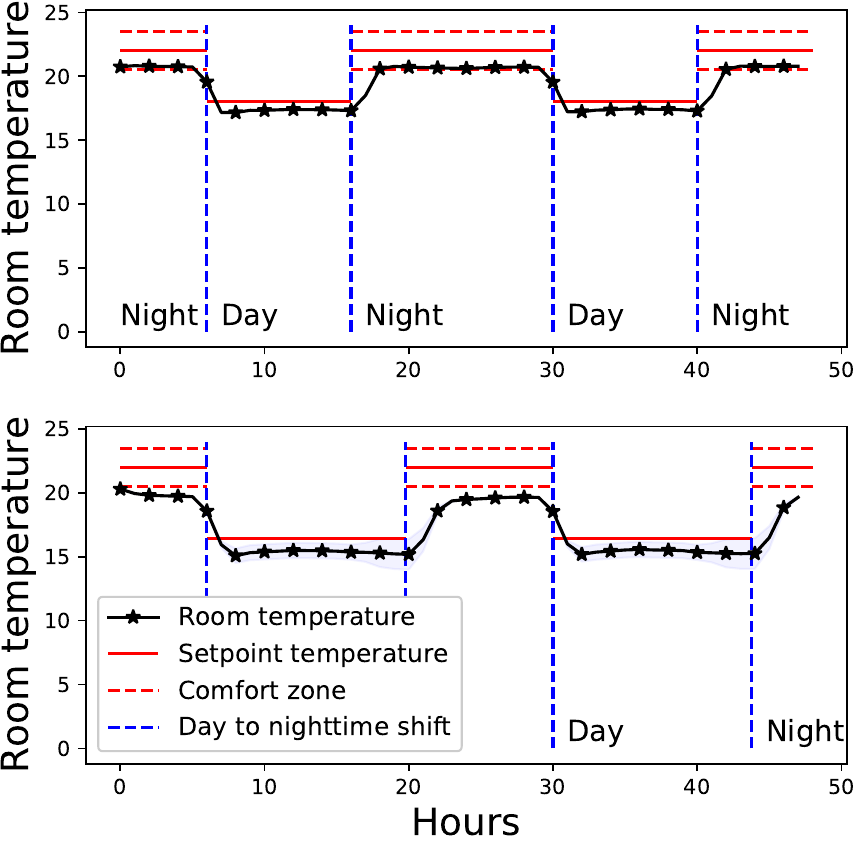}
    \caption{Mean temperature over $48$ hours, where the control inputs are sampled from the weights learned by c.AdaNormalGP (\textit{top}) and GPMW (\textit{bottom}).}
    \label{fig:temp_deviation_of_player_1}
\end{figure}

\begin{figure}[htbp]
    \centering
    \includegraphics[scale = 0.52]{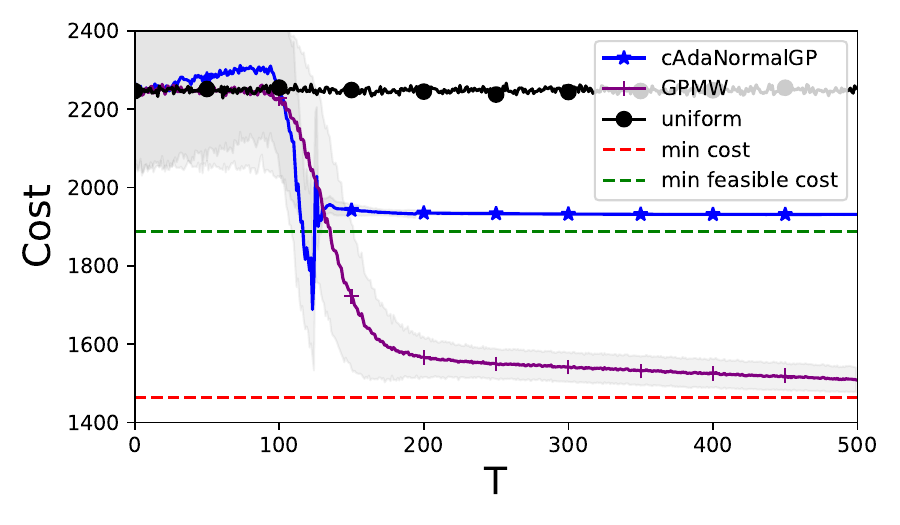}
    \caption{Mean energy cost achieved by c.AdaNormalGP, GPMW, and uniformly at random sampled control inputs for each round $t=1,\ldots T$. The minimum feasible- and the minimum cost are found exhaustively over the entire action space.}
    \label{fig:cost_funcs_of_player_1}
\end{figure}

We compare the performance of c.AdaNormalGP with GPMW \cite{sessa2019no}, which neglects constraints on the maximum temperature deviation from the occupants comfort range. To run c.AdaNormalGP and GPMW, we use the polynomial kernel for the cost function and the squared exponential kernel for the constraint function and run both algorithms for $T=500$ rounds. 
For example for building $1$, Figure \ref{fig:temp_deviation_of_player_1} shows the mean room temperature over $48$ hours resulting from $1000$ control inputs $a_1$ that are sampled according to c.AdaNormalGP (\textit{top}) and GPMW (\textit{bottom}). Following c.AdaNormalGP the mean temperature stays within the occupant's comfort range during the nighttime whereas under GPMW the deviation of the mean temperature exceeds the threshold and thus is below the occupant's comfort range. Furthermore, following c.AdaNormalGP and GPMW each building learns to select control parameters that, given the other buildings energy demand, reduce the building's energy cost compared to the case where the control parameters are selected uniformly at random. This is shown in Figure \ref{fig:cost_funcs_of_player_1} for building $1$. Although the energy cost is lower when following GPMW than when following c.AdaNormalGP such a cost is not attainable when also respecting the constraints on the occupant's comfort temperature range. In conclusion, following c.AdaNormalGP each building not only learns to select feasible control parameters but also attains nearly the minimal cost achievable under the temperature constraints.

\section{Conclusion and further discussion}

We considered the problem of repeated contextual games under unknown constraints. Establishing a connection  between this problem and the sleeping expert problem, we proposed the constrained contextual c.z.AdaNormalGP algorithm, a no-regret no-violation algorithm for playing such games with finite and infinite context spaces. Our algorithm attains high probability regret and cumulative constraint violation bounds, where the former is comparable to bounds obtained in the unconstrained setting. We further showed that if all players follow no-regret, no-violation dynamics then the empirical joint policy converges to a $c.z.CCE$ of the constrained contextual game as the game is repeated infinitely often. We demonstrated the effectiveness of c.z.AdaNormalGP and a multi-building temperature control design problem. Future research directions include learning in repeated games with coupling constraints.
\bibliography{bibliography}

\begin{thebibliography}{}

\bibitem[Abbasi-Yadkori et~al., 2011]{abbasi-yadkori2011}
Abbasi-Yadkori, Y., P\'{a}l, D., and Szepesv\'{a}ri, C. (2011).
\newblock Improved algorithms for linear stochastic bandits.
\newblock In {\em Advances in Neural Information Processing Systems}.

\bibitem[Agrawal and Goyal, 2013]{agrawal2014}
Agrawal, S. and Goyal, N. (2013).
\newblock Thompson sampling for contextual bandits with linear payoffs.
\newblock In {\em Proceedings of the 30th International Conference on Machine Learning}. PMLR.

\bibitem[Ang et~al., 2005]{ang2005pid}
Ang, K.~H., Chong, G., and Li, Y. (2005).
\newblock Pid control system analysis, design, and technology.
\newblock {\em IEEE transactions on control systems technology}.

\bibitem[Blum, 1995]{blum1995}
Blum, A. (1995).
\newblock Empirical support for winnow and weighted-majority based algorithms: results on a calendar scheduling domain.
\newblock In {\em Machine Learning Proceedings 1995}.

\bibitem[Cesa-Bianchi and Lugosi, 2006]{cesa-bianchi06}
Cesa-Bianchi, N. and Lugosi, G. (2006).
\newblock {\em Prediction, Learning, and Games}.
\newblock Cambridge University Press, USA.

\bibitem[Chen et~al., 2022]{chen2022}
Chen, Z., Ma, S., and Zhou, Y. (2022).
\newblock Finding correlated equilibrium of constrained markov game: A primal-dual approach.
\newblock In {\em Advances in Neural Information Processing Systems}.

\bibitem[Chowdhury and Gopalan, 2017]{chowdhury2017kernelized}
Chowdhury, S.~R. and Gopalan, A. (2017).
\newblock On kernelized multi-armed bandits.
\newblock In {\em International Conference on Machine Learning}. PMLR.

\bibitem[Chu et~al., 2011]{chu2011}
Chu, W., Li, L., Reyzin, L., and Schapire, R. (2011).
\newblock Contextual bandits with linear payoff functions.
\newblock In {\em Proceedings of the Fourteenth International Conference on Artificial Intelligence and Statistics}. PMLR.

\bibitem[Clarkson, 2005]{clakson05}
Clarkson, K.~L. (2005).
\newblock Nearest-neighbor searching and metric space dimensions.

\bibitem[Di~Natale et~al., 2022]{dinatale2022}
Di~Natale, L., Lian, Y., Maddalena, E.~T., Shi, J., and Jones, C.~N. (2022).
\newblock Lessons learned from data-driven building control experiments: Contrasting gaussian process-based mpc, bilevel deepc, and deep reinforcement learning.
\newblock In {\em 2022 IEEE 61st Conference on Decision and Control (CDC)}.

\bibitem[Fiducioso et~al., 2019]{fiducioso2019safe}
Fiducioso, M., Curi, S., Schumacher, B., Gwerder, M., and Krause, A. (2019).
\newblock Safe contextual bayesian optimization for sustainable room temperature pid control tuning.
\newblock {\em arXiv preprint arXiv:1906.12086}.

\bibitem[Freund and Schapire, 1995]{freund95}
Freund, Y. and Schapire, R.~E. (1995).
\newblock A decision-theoretic generalization of online learning and an application to boosting.
\newblock In {\em Proceedings of the Second European Conference on Computational Learning Theory}. Springer-Verlag.

\bibitem[Freund et~al., 1997]{freund1997}
Freund, Y., Schapire, R.~E., Singer, Y., and Warmuth, M.~K. (1997).
\newblock Using and combining predictors that specialize.
\newblock In {\em Proceedings of the Twenty-Ninth Annual ACM Symposium on Theory of Computing}. Association for Computing Machinery.

\bibitem[Fritzson and Engelson, 1998]{fritzson1998modelica}
Fritzson, P. and Engelson, V. (1998).
\newblock Modelica—a unified object-oriented language for system modeling and simulation.
\newblock In {\em ECOOP’98—Object-Oriented Programming: 12th European Conference Brussels, Belgium, July 20--24, 1998 Proceedings 12}. Springer.

\bibitem[Guo et~al., 2022]{guo2022}
Guo, H., Liu, X., Wei, H., and Ying, L. (2022).
\newblock Online convex optimization with hard constraints: Towards the best of two worlds and beyond.
\newblock In {\em Advances in Neural Information Processing Systems}.

\bibitem[Hall et~al., 2022]{hall2022receding}
Hall, S., Belgioioso, G., Liao-McPherson, D., and Dorfler, F. (2022).
\newblock Receding horizon games with coupling constraints for demand-side management.
\newblock In {\em 2022 IEEE 61st Conference on Decision and Control (CDC)}. IEEE.

\bibitem[Hart and Mas-Colell, 2000]{hart2000simple}
Hart, S. and Mas-Colell, A. (2000).
\newblock A simple adaptive procedure leading to correlated equilibrium.
\newblock {\em Econometrica}.

\bibitem[Hazan and Megiddo, 2007]{hazan07}
Hazan, E. and Megiddo, N. (2007).
\newblock Online learning with prior knowledge.
\newblock In {\em Proceedings of the 20th Annual Conference on Learning Theory}, COLT'07. Springer-Verlag.

\bibitem[Krause and Ong, 2011]{ong2011}
Krause, A. and Ong, C. (2011).
\newblock Contextual gaussian process bandit optimization.
\newblock In {\em Advances in Neural Information Processing Systems}.

\bibitem[Krauthgamer and Lee, 2004]{krauthgamer2004}
Krauthgamer, R. and Lee, J.~R. (2004).
\newblock Navigating nets: Simple algorithms for proximity search.
\newblock In {\em Proceedings of the Fifteenth Annual ACM-SIAM Symposium on Discrete Algorithms}. Society for Industrial and Applied Mathematics.

\bibitem[Laustsen, 2008]{laustsen2008energy}
Laustsen, J. (2008).
\newblock Energy efficiency requirements in building codes, energy efficiency policies for new buildings. iea information paper.

\bibitem[Li et~al., 2010]{li2010}
Li, L., Chu, W., Langford, J., and Schapire, R.~E. (2010).
\newblock A contextual-bandit approach to personalized news article recommendation.
\newblock In {\em Proceedings of the 19th International Conference on World Wide Web}. Association for Computing Machinery.

\bibitem[Lin et~al., 2022]{lin2022}
Lin, Y., Wang, Y., and Zhou, E. (2022).
\newblock Risk-averse contextual multi-armed bandit problem with linear payoffs.
\newblock {\em Journal of Systems Science and Systems Engineering}.

\bibitem[Littlestone and Warmuth, 1994]{littlestone1994}
Littlestone, N. and Warmuth, M. (1994).
\newblock The weighted majority algorithm.
\newblock {\em Information and Computation}.

\bibitem[Luo and Schapire, 2014]{luo14a}
Luo, H. and Schapire, R.~E. (2014).
\newblock A drifting-games analysis for online learning and applications to boosting.
\newblock In {\em Advances in Neural Information Processing Systems}.

\bibitem[Luo and Schapire, 2015]{luo15}
Luo, H. and Schapire, R.~E. (2015).
\newblock Achieving all with no parameters: Ada{N}ormal{H}edge.
\newblock In {\em Proceedings of The 28th Conference on Learning Theory}. PMLR.

\bibitem[Mourtada and Ga\"{\i}ffas, 2019]{mourtada2019}
Mourtada, J. and Ga\"{\i}ffas, S. (2019).
\newblock On the optimality of the hedge algorithm in the stochastic regime.
\newblock {\em J. Mach. Learn. Res.}

\bibitem[Rasmussen and Williams, 2005]{rasmussen2005}
Rasmussen, C.~E. and Williams, C. K.~I. (2005).
\newblock {\em Gaussian Processes for Machine Learning (Adaptive Computation and Machine Learning)}.
\newblock The MIT Press.

\bibitem[Scharnhorst et~al., 2021]{scharnhorst2021energym}
Scharnhorst, P., Schubnel, B., Fern{\'a}ndez~Bandera, C., Salom, J., Taddeo, P., Boegli, M., Gorecki, T., Stauffer, Y., Peppas, A., and Politi, C. (2021).
\newblock Energym: A building model library for controller benchmarking.
\newblock {\em Applied Sciences}.

\bibitem[Sessa et~al., 2019]{sessa2019no}
Sessa, P.~G., Bogunovic, I., Kamgarpour, M., and Krause, A. (2019).
\newblock No-regret learning in unknown games with correlated payoffs.
\newblock {\em Advances in Neural Information Processing Systems}.

\bibitem[Sessa et~al., 2020]{sessa21}
Sessa, P.~G., Bogunovic, I., Krause, A., and Kamgarpour, M. (2020).
\newblock Contextual games: Multi-agent learning with side information.
\newblock In {\em Advances in Neural Information Processing Systems}.

\bibitem[Slivkins, 2014]{slivkins2014a}
Slivkins, A. (2014).
\newblock Contextual bandits with similarity information.
\newblock {\em Journal of Machine Learning Research}.

\bibitem[Srinivas et~al., 2010]{srinivas2009gaussian}
Srinivas, N., Krause, A., Kakade, S., and Seeger, M. (2010).
\newblock Gaussian process optimization in the bandit setting: No regret and experimental design.
\newblock In {\em In Proceedings International Conference on Machine Learning (ICML)}.

\bibitem[Sui et~al., 2015]{sui15}
Sui, Y., Gotovos, A., Burdick, J., and Krause, A. (2015).
\newblock Safe exploration for optimization with gaussian processes.
\newblock In {\em Proceedings of the 32nd International Conference on Machine Learning}. PMLR.

\bibitem[Sun et~al., 2017]{sun2017a}
Sun, W., Dey, D., and Kapoor, A. (2017).
\newblock Safety-aware algorithms for adversarial contextual bandit.
\newblock In {\em Proceedings of the 34th International Conference on Machine Learning}. PMLR.

\bibitem[Turchetta et~al., 2019]{turchetta2019safe}
Turchetta, M., Berkenkamp, F., and Krause, A. (2019).
\newblock Safe exploration for interactive machine learning.
\newblock In {\em Advances in Neural Information Processing Systems}.

\bibitem[Usmanova et~al., 2020]{usmanova21}
Usmanova, I., Krause, A., and Kamgarpour, M. (2020).
\newblock Safe non-smooth black-box optimization with application to policy search.
\newblock In {\em Proceedings of the 2nd Conference on Learning for Dynamics and Control}, Proceedings of Machine Learning Research. PMLR.

\bibitem[Valko et~al., 2013]{valko2013}
Valko, M., Korda, N., Munos, R., Flaounas, I., and Cristianini, N. (2013).
\newblock Finite-time analysis of kernelised contextual bandits.
\newblock In {\em Proceedings of the Twenty-Ninth Conference on Uncertainty in Artificial Intelligence}. AUAI Press.

\bibitem[Wei et~al., 2022]{Wei_Liu_Ying_2022}
Wei, H., Liu, X., and Ying, L. (2022).
\newblock A provably-efficient model-free algorithm for infinite-horizon average-reward constrained markov decision processes.
\newblock {\em Proceedings of the AAAI Conference on Artificial Intelligence}.

\bibitem[Xu et~al., 2023]{jones22}
Xu, W., Jiang, Y., Svetozarevic, B., and Jones, C. (2023).
\newblock Constrained efficient global optimization of expensive black-box functions.
\newblock In {\em International Conference on Machine Learning}. PMLR.

\bibitem[Yi et~al., 2021]{yi2021}
Yi, X., Li, X., Yang, T., Xie, L., Chai, T., and Johansson, K.~H. (2021).
\newblock Regret and cumulative constraint violation analysis for online convex optimization with long term constraints.

\bibitem[Yu and Neely, 2020]{Yu2020}
Yu, H. and Neely, M.~J. (2020).
\newblock A low complexity algorithm with $\mathcal{O}(\sqrt{T})$ regret and $\mathcal{O}(1)$ constraint violations for online convex optimization with long term constraints.
\newblock {\em Journal of Machine Learning Research}.

\bibitem[Zhou and Ji, 2022]{zhou22}
Zhou, X. and Ji, B. (2022).
\newblock On kernelized multi-armed bandits with constraints.
\newblock In {\em Advances in Neural Information Processing Systems}.

\end{thebibliography}


\begin{thebibliography}{}
\setlength{\itemindent}{-\leftmargin}
\makeatletter\renewcommand{\@biblabel}[1]{}\makeatother
\bibitem{} J.~Alspector, B.~Gupta, and R.~B.~Allen (1989).
    \newblock Performance of a stochastic learning microchip.
    \newblock In D. S. Touretzky (ed.),
    \textit{Advances in Neural Information Processing Systems 1}, 748--760.
    San Mateo, Calif.: Morgan Kaufmann.

\bibitem{} F.~Rosenblatt (1962).
    \newblock \textit{Principles of Neurodynamics.}
    \newblock Washington, D.C.: Spartan Books.

\bibitem{} G.~Tesauro (1989).
    \newblock Neurogammon wins computer Olympiad.
    \newblock \textit{Neural Computation} \textbf{1}(3):321--323.
\end{thebibliography}

\newpage
\appendix
\onecolumn

\section{Appendix}

\subsection{Supplementary Material for Section \ref{sec:c.z.AdaNormalGP}}

For proving Theorems \ref{thm:finite_Z}, \ref{thm:infinite_Z}, \ref{thm:general_algo_finite_Z} and \ref{thm:general_algo_infinite_Z} we make use of the following lemma and theorem.

\begin{lemma}\cite[Theorem~2]{chowdhury2017kernelized}\label{lem:confidence}
Let $\mathcal{H}_k$ be the RKHS of real-valued functions on $\mathcal{X}\subset\mathbb{R}^d$ with underlying kernel function $k$. Consider an unknown function $f:\mathcal{X}\rightarrow\mathbb{R}$ in $\mathcal{H}_k$ such that $\|f\|_k\leq D$, and the sampling model $y^t=f(x^t)+\epsilon^t$, where $\epsilon^t$ is $\sigma$-sub-Gaussian (with independence between times). By setting
\begin{align*}
    \beta^t=D+\sigma\sqrt{2(\gamma^{t-1}+1+\log(1/\delta))}
\end{align*}
the following holds with probability at least $1-\delta$:
\begin{align*}
    |\mu^{t-1}(x)-f(x)|\leq\beta^t\sigma^{t-1}(x),\hspace{0.3cm}\forall x\in\mathcal{X},\ \forall t\geq 1,
\end{align*}
where $\mu^{t-1}(\cdot)$ and $\sigma^{t-1}(\cdot)$ are given in \eqref{eq:mu}-\eqref{eq:var}. Here, $\gamma^{t-1}$ denotes the \textit{maximum information gain}, a kernel-dependent quantity defined in Section \ref{sec:c.z.AdaNormalGP}.
\end{lemma}

\begin{theorem}\cite[Theorem~3]{luo15}\label{thm:pseudo_ada}
    For the sleeping expert problem, the regret of AdaNormalHedge\footnote{We consider AdaNormalHedge (Algorithm 1 in \cite{luo15}), where the $p^t_{a_i}$ is predicted via equation (3) in \cite{luo15}.} is bounded as follows:
    \begin{align}
    \begin{split}
        R^T(a_i)=\sum_{t=1}^T \mathbbm{1}_{a_i,t}(\mathbb{E}[l^t]- l^t(a_i))
        &\leq \sqrt{3C_{a_i}^T(\log(K)+\log(B)+\log(1+\log(K)))},\\
        &\leq \sqrt{3T(\log(K)+\log(B)+\log(1+\log(K)))}
        \end{split}
    \end{align}
    where $l^t(a_i)$ is the loss of action (expert) $a_i\in\mathcal{A}_i$, {\small$l^t=[l^t(a_i)]_{a_i\in\mathcal{A}_i}\in[0,1]^K$} is the loss vector of all actions (experts), {\small$C_{a_i}^T=\sum_{t=1}^T |\mathbb{E}[l^t(a_i^t)]-l^t(a_i)|$} and {\small$B=1+\frac{3}{2}\frac{1}{K}\sum_{a_i\in\mathcal{A}_i} (1+\log(1+C_{a_i}^T))\leq \frac{5}{2}+\frac{3}{2}\log(1+T)$}. Here, $\mathbbm{1}_{a_i,t}$ denotes the indicator function defined as:
    \begin{align*}
        \mathbbm{1}_{a_i,t}=\begin{cases}
        1, & \text{if action (expert) $a_i$ is available (awake) at round $t$},\\
        0, & \text{otherwise}.
        \end{cases}
    \end{align*}
\end{theorem}

\subsubsection{Finite number of contexts: Proof of Theorem \ref{thm:finite_Z}}\label{app:finite_Z}

In this section, we provide a proof for Theorem \ref{thm:finite_Z}. By making use of Lemma \ref{lem:confidence}, we first show that upper and lower confidence bounds hold simultaneously for the reward and all constraint functions. Next, we show that c.z.AdaNormalGP does not declare infeasibility with high probability if the underlying game is feasible. We then prove an upper bound on the cumulative constraint violations making use of the previous two results and techniques from GP optimization \cite{srinivas2009gaussian} to account for not knowing the true constraint functions. To obtain an upper bound on the regret we leverage that for each context $z\in\mathcal{Z}$ a separate distribution is maintained and analyze the regret obtained by each context $z\in\mathcal{Z}$. Fixing a context $z\in\mathcal{Z}$, the proof further follows by decomposing the regret into the sum of two terms. The first term corresponds to the regret that a player incurs with respect to the optimistic reward estimates, i.e., the upper confidence bounds, when the set of feasible actions is a priori unknown. This term can be upper bounded by Theorem \ref{thm:pseudo_ada} thanks to the connection between constrained games and the sleeping expert problem. The second term stems from not knowing the true reward function and can be upper bounded by using techniques from GP optimization.

\begin{lemma}\label{lem:union_confidence}
    With probability $1-\frac{\delta}{2}$ the following holds simultaneously:
    \begin{align}\label{eq:simultaneous_UCB_LCB}
        \begin{split}
            &\text{LCB}_0^t(a,z)\leq r(a,z)\leq \min\{1,\text{UCB}_0^t(a,z)\},\quad \forall a\in\mathcal{A}, \forall z\in\mathcal{Z}, \forall t\geq 1,\\
            & \text{LCB}_m^t(a_i,z)\leq g_{m}(a_i,z)\leq \text{UCB}_m^t(a_i,z),\quad \forall a_i\in\mathcal{A}_i,\forall z\in\mathcal{Z}, \forall t\geq 1, \forall m\in[M],
        \end{split}
    \end{align}
    with {\small$\beta_m^t=B_m+\sigma\sqrt{2(\gamma_m^{t-1}+1+\log(2(M+1)/\delta))}$} for {\small$m\in\{0\}\cup[M]$}. 
\end{lemma}

\begin{proof}
By Lemma \ref{lem:confidence} and since the true unknown reward function $r$ lies in $[0,1]$, with probability at least {\small$1-\frac{\delta}{2(M+1)}$} the reward function can be upper and lower bounded as follows:
    \begin{align}
        \begin{split}\label{eq:confidence_reward}
            &\text{LCB}_0^t(a,z)\leq r(a,z)\leq \min\{1,\text{UCB}_0^t(a,z)\},\hspace{0.3cm} \forall a\in\mathcal{A},\ \forall z\in\mathcal{Z}, \ \forall t\geq 1,
        \end{split}
    \end{align}
    where the upper- and lower confidence bounds $\text{UCB}_0^t(a,z)$ and $\text{LCB}_0^t(a,z)$ are defined as in \eqref{eq:ucb}-\eqref{eq:lcb} and {\small$\beta_0^t=B_0+\sigma\sqrt{2(\gamma_0^{t-1}+1+\log(2(M+1)/\delta))}$}. Analogously, with probability at least {\small$1-\frac{\delta}{2(M+1)}$} each unknown constraint function $g_{m}$, $m\in[M]$, can be upper and lower bounded as follows:
    \begin{align}
        \begin{split}\label{eq:confidence_constraint}
            \text{LCB}_m^t(a_i,z)\leq g_{m}(a_i,z)\leq \text{UCB}_m^t(a_i,z),\hspace{0.3cm} \forall a_i\in\mathcal{A}_i,\ \forall z\in\mathcal{Z},\ \forall t\geq 1,
        \end{split}
    \end{align}
    where $\text{UCB}_m^t(a_i,z)$ and $\text{LCB}_m^t(a_i,z)$ are defined as in \eqref{eq:ucb}-\eqref{eq:lcb} and {\small$\beta_m^t=B_m+\sigma\sqrt{2(\gamma_m^{t-1}+1+\log(2(M+1)/\delta))}$}.

    Note that the confidence bounds $\text{LCB}_0^t(a,z)$, $\text{UCB}_0^t(a,z)$, $\text{LCB}_m^t(a_i,z)$ and $\text{UCB}_m^t(a_i,z)$, $m\in[M]$, are random variables since they depend on the random sampling process of the algorithm (line 7 of Algorithm \ref{alg:contextual_constrainedGPHedge}) and on noisy observations of the reward $\Tilde{r}^t$ and constraints $\Tilde{g}_m^t$, $m\in[M]$. Define {\small$E_0:=\cap_{a\in\mathcal{A}}\cap_{z\in\mathcal{Z}}\cap_{t\in[T]}\{\text{LCB}_0^t(a,z)\leq r(a,z)\leq \min\{1,\text{UCB}_0^t(a,z)\}\}$} and define {\small$E_m:=\cap_{a_i\in\mathcal{A}_i}\cap_{z\in\mathcal{Z}}\cap_{t\in[T]}\{\text{LCB}_m^t(a_i,z)\leq g_{m}(a_i,z)\leq \text{UCB}_m^t(a_i,z)\}$} for all $m\in[M]$. Then, using De Morgan's law, the union bound and \eqref{eq:confidence_reward}-\eqref{eq:confidence_constraint} we have
    \begin{align*}
        \mathbb{P}(\cap_{m=0}^M E_m)
        &=1-\mathbb{P}((\cap_{m=0}^M E_m)^C)
        =1-\mathbb{P}(\cup_{m=0}^M E_m^C)\\
        &\geq 1-\sum_{m=0}^M\mathbb{P}(E_m^C)\\
        &\geq 1-\sum_{m=0}^M\frac{\delta}{2(M+1)}
         =1-\frac{\delta}{2}.
    \end{align*}
\end{proof}

\begin{lemma}\label{lem:infeasibility_not_declared}
    With probability $1-\frac{\delta}{2}$, c.z.AdaNormalGP (line 3 of Algorithm \ref{alg:contextual_constrainedGPHedge}) does not declare infeasibility.
\end{lemma}

\begin{proof}
        
    By Assumption \ref{ass:feasibility_context} the problem given in \eqref{problem_context} is feasible and has optimal solution $\overline{\pi}_i^*$, i.e., 
    \begin{align*}
        \overline{\pi}_i^*\in\arg\max_{\pi_i\in\Pi_i} &\sum_{t=1}^T r(\pi_i(z^t),a_{-i}^t,z^t)\\
        \text{s.t.\ }& g_{m}(\pi_i(z^t),z^t)\leq 0, \quad\forall m\in[M],\forall t\geq 1.
    \end{align*}
    By Lemma \ref{lem:union_confidence}, with probability $1-\frac{\delta}{2}$, $\text{LCB}_m^t(a_i,z)\leq g_{m}(a_i,z)$ for all $a_i\in\mathcal{A}_i$, $z\in\mathcal{Z}$, $m\in[M]$, and $t\geq 1$. Thus, by setting $a_i=\overline{\pi}_i^*(z^t)$ with probability $1-\frac{\delta}{2}$ we have that $\text{LCB}_m^t(\overline{\pi}_i^*(z^t),z^t)\leq 0$ for all $m\in[M]$ and $t\geq 1$ and infeasibility is not declared.
\end{proof}

We proceed with upper-bounding the cumulative constraint violations for each constraint.

\begin{lemma}\label{lem:cum_violation_bound}
    With probability at least {\small$1-\frac{\delta}{2}$}  the cumulative constraint violations $\mathcal{V}_m^T$ for each constraint $g_m$, $m\in[M]$, are upper-bounded as follows:
    \begin{align*}
        \mathcal{V}_{m}^T =\mathcal{O}\bigg(\beta_m^T\sqrt{T\gamma_m^T}\bigg),\quad \forall m\in[M].
    \end{align*}
\end{lemma}

\begin{proof}
    With probability at least $1-\frac{\delta}{2}$ the following holds:
    \begin{align}
            \mathcal{V}_{m}^T &= \sum_{t=1}^T [g_{m}(a_i^t,z^t)]_+\\
            &= \sum_{t=1}^T [g_{m}(a_i^t,z^t)-\text{LCB}_m^t(a_i^t,z^t)+\text{LCB}_m^t(a_i^t,z^t)]_+\\
            &\leq \sum_{t=1}^T [g_{m}(a_i^t,z^t)-\text{LCB}_m^t(a_i^t,z^t)]_+ +\sum_{t=1}^T[\text{LCB}_m^t(a_i^t,z^t)]_+\\ 
            &\leq \sum_{t=1}^T [\text{UCB}_m^t(a_i^t,z^t)-\text{LCB}_m^t(a_i^t,z^t)]_+ +\sum_{t=1}^T[\underbrace{\text{LCB}_m^t(a_i^t,z^t)}_{\leq 0}]_+ \label{eq:use_claim1}\\ 
            &\leq\sum_{t=1}^T [\text{UCB}_m^t(a_i^t,z^t)-\text{LCB}_m^t(a_i^t,z^t)]_+\label{eq:0whp}\\
            &\leq 2\beta_m^T \sum_{t=1}^T \sigma_m^{t-1}(a_i^t,z^t)\label{eq:beta_increasing,0whp}\\
            &\leq C_1\beta_m^T\sqrt{T\gamma_m^T}\label{eq:srinivas},
    \end{align}
    where equation \eqref{eq:use_claim1} follows from Lemma \ref{lem:union_confidence} and holds with probability at least {\small$1-\frac{\delta}{2}$}. Equation \eqref{eq:0whp} follows from the fact that the sampled action $a_i^t$ at each round $t$ is feasible with respect to the lower confidence bound $\text{LCB}_m^t(\cdot)$ for all $m\in[M]$ (line  7 of Algorithm \ref{alg:contextual_constrainedGPHedge}) and that with probability at least {\small$1-\frac{\delta}{2}$} infeasibility is not declared. Equation \eqref{eq:beta_increasing,0whp} follows firstly from the definition of $\text{UCB}_m^t(\cdot)$ in \eqref{eq:ucb} and $\text{UCB}_m^t(\cdot)$ in \eqref{eq:lcb} and the fact that $\beta_m^t$ is increasing in $t$.The last equation makes use of \cite[Lemma~5.4]{srinivas2009gaussian} which provides an upper bound on the sum of posterior standard deviations, where {\small$C_1=\frac{8}{\log(1+\sigma^{-2})}$}.
\end{proof}
    
Next, we proceed with upper-bounding the constrained contextual regret, short regret. 

\begin{lemma}\label{lem:regret_bound_finite_Z}
    Conditioned on equation \eqref{eq:simultaneous_UCB_LCB} in Lemma \ref{lem:union_confidence} holding true,  the constrained contextual regret $R^T$ is upper-bounded with probability $1-\frac{\delta}{2}$ by:
    \begin{align*}
        R^T=\mathcal{O}\bigg(\sqrt{|\mathcal{Z}|T(\log(K)+\log(B)+\log(1+\log(K))}+\sqrt{T\log(2/\delta)}+\beta_0^T\sqrt{T\gamma_0^T}\bigg).
    \end{align*}
\end{lemma}

\begin{proof}
    
Recall the definition of regret:
    \begin{align*}
\begin{split}
    R^T=&\max_{\substack{\pi_i\in\Pi_i}} \sum_{t=1}^T r(\pi_i(z^t),a_{-i}^t,z^t)-\sum_{t=1}^T r(a_i^t,a_{-i}^t,z^t)\\
    &\text{s.t.\ }g_{m}(\pi_i(z^t),z^t)\leq 0,\quad\forall m\in[M], \forall t\geq 1,
\end{split}\\
    =&\sum_{t=1}^T r(\overline{\pi}_i^*(z^t),a_{-i}^t,z^t)-\sum_{t=1}^T r(a_i^t,a_{-i}^t,z^t).
\end{align*}

    Note that for a finite context space $\mathcal{Z}$, c.z.AdaNormalGP maintains a separate distribution $p^t(z)$ for each context $z\in\mathcal{Z}$, where $p^t(z^t)$ is computed according to Algorithm \ref{alg:strategy_finite_Z}. Thus, the regret notion can be rewritten as:
    \begin{align*}
        R^T&=\sum_{z\in\mathcal{Z}}\sum_{t:z^t=z} r(\overline{\pi}_i^*(z^t),a_{-i}^t,z^t)-r(a_i^t,a_{-i}^t,z^t).
    \end{align*}

    The reward and thus $R^T$ can be further upper bounded as follows:
    \begin{align}
        R^T&= \sum_{z\in\mathcal{Z}}\sum_{t:z^t=z}r(\overline{\pi}_i^*(z^t),a_{-i}^t,z^t)-r(a_i^t,a_{-i}^t,z^t)\nonumber\\
        &\leq\sum_{z\in\mathcal{Z}}\sum_{t:z^t=z} \min\{1,\text{UCB}_0^t(\overline{\pi}_i^*(z^t),a_{-i}^t,z^t)\}- (\text{LCB}_0^t(a_i^t,a_{-i}^t,z^t))\label{eq:finite_Z,UCB,LCB}\\
        &\leq\sum_{z\in\mathcal{Z}}\sum_{t:z^t=z} \min\{1,\text{UCB}_0^t(\overline{\pi}_i^*(z^t),a_{-i}^t,z^t)\}- (\text{UCB}_0^t(a_i^t,a_{-i}^t,z^t)-2\beta_0^t\sigma_0^{t-1}(a_i^t,a_{-i}^t,z^t))\nonumber\\
        &\leq \sum_{z\in\mathcal{Z}}\sum_{t:z^t=z} \min\{1,\text{UCB}_0^t(\overline{\pi}_i^*(z^t),a_{-i}^t,z^t)\}- \min\{1,\text{UCB}_0^t(a_i^t,a_{-i}^t,z^t)\}\nonumber\\
        &+\sum_{t=1}^T 2\beta_0^t\sigma_0^{t-1}(a_i^t,a_{-i}^t,z^t)\nonumber\\
        &\leq \sum_{z\in\mathcal{Z}}\sum_{t:z^t=z} (\min\{1,\text{UCB}_0^t(\overline{\pi}_i^*(z^t),a_{-i}^t,z^t)\}- \min\{1,\text{UCB}_0^t(a_i^t,a_{-i}^t,z^t)\})+C_1\beta_0^T\sqrt{T\gamma_0^T}\label{eq:finite_Z,beta,bound},
    \end{align}
    where in equation \eqref{eq:finite_Z,UCB,LCB} we used equation \eqref{eq:simultaneous_UCB_LCB} from Lemma \ref{lem:union_confidence} and in equation \eqref{eq:finite_Z,beta,bound} we again used \cite[Lemma 5.4]{srinivas2009gaussian} with $C_1=\frac{8}{\log(1+\sigma^{-2})}$.

    To further upper-bound $R^T$ we rely on the Hoeffding-Azuma inequality \cite[Lemma~A.7]{cesa-bianchi06}. Note that the collection of random variables $\{V^t\}_{t=1}^T$ defined as $V^t:=\min\{1,\text{UCB}_0^t(a_i^t,a_{-i}^t,z^t)\}-\sum_{a_i\in\mathcal{A}_i} \overline{p}_{a_i}^t(z^t)\min\{1,\text{UCB}_0^t(a_i,a_{-i}^t,z^t)\}$ forms a martingale difference sequence\footnote{{\tiny$\{V^t\}_{t=1}^T$} forms a martingale difference sequence with respect to {\tiny$U^1,\ldots,U^T$}, where {\tiny$U^t\in\big[\sum_{j=1}^{k-1}p_j^t(z^t),\sum_{j=1}^kp_j^t(z^t)\big)$} if and only if $a_i^t=k$.}. Then, the Hoeffding-Azuma inequality is applicable and the following holds with probability at least $1-\delta/2$:
    {\small
    \begin{align}
        \sum_{z\in\mathcal{Z}}\sum_{t:z^t=z}\big|\min\{1,\text{UCB}_0^t(a_i^t,a_{-i}^t,z^t)\}-\sum_{a_i\in\mathcal{A}_i} \overline{p}_{a_i}^t(z^t)\min\{1,\text{UCB}_0^t(a_i,a_{-i}^t,z^t)\}\big|\leq\sqrt{T/2\log(2/\delta)}.\label{eq:finite_Z_hoeffding}
    \end{align}}
    Here, we used that $V^t\in[0,1]$ and $\sum_{z\in\mathcal{Z}}\sum_{t:z^t=z}1= T$. Plugging equation \eqref{eq:finite_Z_hoeffding} into equation \eqref{eq:finite_Z,beta,bound}, with probability at least $1-\delta/2$, we obtain the following upper bound on $R^T$:
    \begin{align}
    \begin{split}\label{eq:finite_Z_AdaNormalHedgeRegret}
        R^T&\leq  \sum_{z\in\mathcal{Z}}\sum_{t:z^t=z} \big(\min\{1,\text{UCB}_0^t(\overline{\pi}_i^*(z^t),a_{-i}^t,z^t)\}- \sum_{a_i\in\mathcal{A}_i} \overline{p}_{a_i}^t(z^t)\min\{1,\text{UCB}_0^t(a_i,a_{-i}^t,z^t)\}\big)\\
        &+\sqrt{T/2\log(2/\delta)}+C_1\beta_0^T\sqrt{T\gamma_0^T}.
    \end{split}
    \end{align}
    Define the function $f^t(\cdot)=\min\{1,\text{UCB}_0^t(\cdot,a_{-i}^t,z^t)\}$. Since we conditioned on equation \eqref{eq:simultaneous_UCB_LCB} in Lemma \ref{lem:union_confidence} holding true, $\text{UCB}_0^t(\cdot)\geq 0$ since $r(\cdot)\geq 0$ and thus $f^t(\cdot)\in[0,1]^K$. For a fixed $z\in\mathcal{Z}$, the first term in equation \eqref{eq:finite_Z_AdaNormalHedgeRegret} can be rewritten as:
    \begin{align}
    &\sum_{t:z^t=z} \min\{1,\text{UCB}_0^t(\overline{\pi}_i^*(z^t),a_{-i}^t,z^t)\}- \sum_{a_i\in\mathcal{A}_i}\overline{p}_{a_i}^t(z^t)\min\{1,\text{UCB}_0^t(a_i,a_{-i}^t,z^t)\}\nonumber\\
        =& \sum_{t:z^t=z} f^t(\overline{\pi}_i^*(z^t))-\sum_{a_i\in\mathcal{A}_i}\overline{p}_{a_i}^t(z^t)f^t(a_i^t)\label{eq:finite_Z_normal_regret}\\
        =& \sum_{t:z^t=z} \mathbbm{1}_{\overline{\pi}_i^*(z^t),t}\bigg(f^t(\overline{\pi}_i^*(z^t))-\sum_{a_i\in\mathcal{A}_i}\overline{p}_{a_i}^t(z^t)f^t(a_i^t)\bigg),\label{eq:finite_Z_sleeping_regret}
    \end{align}
    where for all $a_i\in\mathcal{A}_i$ we define:
    \begin{align}\label{eq:finite_Z_indicator_sleeping}
        \begin{split}
            \mathbbm{1}_{a_i,t}:=\begin{cases}
			1, & \text{if $\text{LCB}_m^t(a_i,z^t)\leq 0$ for every $m\in[M],$}\\
            0, & \text{otherwise}.
		 \end{cases}
        \end{split}
    \end{align}
    Since we conditioned on equation \eqref{eq:simultaneous_UCB_LCB} in Lemma \ref{lem:union_confidence} holding true, $\mathbbm{1}_{\overline{\pi}_i^*(z^t),t}$ equals $1$ for all $t\geq 1$ and therefore equation \eqref{eq:finite_Z_normal_regret} and equation \eqref{eq:finite_Z_sleeping_regret} are equivalent. Now observe that equation \eqref{eq:finite_Z_sleeping_regret} is precisely the regret which a player with reward function $f^t(\cdot)\in[0, 1]$ incurs in a sleeping experts problem after $T_z=\sum_{t=1}^T 1_{\{z^t=z
    \}}$ repetitions of the game. Here $T_z$ denotes the number of times context $z^t$ is revealed. As mentioned previously, for each context $z\in\mathcal{Z}$ a distribution $p^t(z)$ is maintained which is updated whenever the revealed context $z^t$ equals $z$. The update rule for each distribution (Line $3-4$ in Algorithm \ref{alg:strategy_finite_Z}) corresponds exactly to that of general AdaNormalHedge proposed in \cite{luo14a}. Thus, at each round $t$ action $a_i^t$ is chosen according to general AdaNormalHedge with the set of sleeping experts $[\mathbbm{1}_{a_i,t}]_{a_i\in\mathcal{A}_i}$ defined as in equation \eqref{eq:finite_Z_indicator_sleeping} and which receives the full information feedback $\hat{r}^t=[f^t(a_i)]_{a_i\in\mathcal{A}_i}$. Concretely, $a_i^t$ is sampled from a subset of awake experts, respective  available actions, where an action $a_i\in\mathcal{A}_i$ is available at round $t$, if $\mathbbm{1}_{a_i,t}=1$, and unavailable, if $\mathbbm{1}_{a_i,t}=0$.

    In summary, equation \eqref{eq:finite_Z_sleeping_regret} corresponds to the regret incurred by the general AdaNormalHedge algorithm after $T_z$ repetitions and can therefore be upper-bounded by Theorem \ref{thm:pseudo_ada}:
    {\small
    \begin{align}\label{eq:finite_Z_sleeping_regret_bound}
        \sum_{t:z^t=z} \mathbbm{1}_{\overline{\pi}_i^*(z^t),t}\bigg(f^t(\overline{\pi}_i^*(z^t))-\sum_{a_i\in\mathcal{A}_i}\overline{p}_{a_i}^t(z^t)f^t(a_i^t)\bigg)\leq\sqrt{3T_z(\log(K)+\log(B)+\log(1+\log(K)))},
    \end{align}}
    where $B\leq \frac{5}{2}+\frac{3}{2}\log(1+T_z)$. 
    
    Now summing over all contexts $z\in\mathcal{Z}$ and using the Cauchy-Schwarz inequality, we can further upper-bound equation \eqref{eq:finite_Z_AdaNormalHedgeRegret} with probability at least $1-\delta/2$:
    \begin{align}
        R^T\leq&\sum_{z\in\mathcal{Z}}\sum_{t:z^t=z} \min\{1,\text{UCB}_0^t(\overline{\pi}_i^*(z^t),a_{-i}^t,z^t)\}- \sum_{a_i\in\mathcal{A}_i}\overline{p}_{a_i}^t(z^t)\min\{1,\text{UCB}_0^t(a_i,a_{-i}^t,z^t)\}\nonumber\\
        &+\sqrt{T/2\log(2/\delta)}+C_1\beta_0^T\sqrt{T\gamma_0^T}\nonumber\\
    &\leq \sum_{z\in\mathcal{Z}}\sqrt{3T_z(\log(K)+\log(B)+\log(1+\log(K)))}+\sqrt{T/2\log(2/\delta)}+C_1\beta_0^T\sqrt{T\gamma_0^T}\nonumber\\
    &\leq \sqrt{|\mathcal{Z}|\sum_{z\in\mathcal{Z}}3T_z(\log(K)+\log(B)+\log(1+\log(K)))}+\sqrt{T/2\log(2/\delta)}+C_1\beta_0^T\sqrt{T\gamma_0^T}\nonumber\\
    &=\sqrt{|\mathcal{Z}|3T(\log(K)+\log(B)+\log(1+\log(K)))}+\sqrt{T/2\log(2/\delta)}+C_1\beta_0^T\sqrt{T\gamma_0^T}\nonumber
\end{align}

\end{proof}

\begin{proof}(of Theorem \ref{thm:finite_Z})
    By combining Lemma \ref{lem:union_confidence}-\ref{lem:regret_bound_finite_Z} Theorem \ref{thm:finite_Z} follows. In particular it follows by standard probability arguments:
\begin{align*}
    \mathbb{P}(\Tilde{E}_1\cap \Tilde{E}_2)
    &= 1-\mathbb{P}((\Tilde{E}_1\cap \Tilde{E}_2)^C)\\
    &=1-\mathbb{P}(\Tilde{E}_1^C\cup \Tilde{E}_2^C)\\
    &\geq 1-(\mathbb{P}(\Tilde{E}_1^C)+\mathbb{P}(\Tilde{E}_2^C))\\
    &\geq 1-(\frac{\delta}{2}+\frac{\delta}{2})
    = 1-\delta,
\end{align*}
where $\Tilde{E}_1$ corresponds to equation \eqref{eq:simultaneous_UCB_LCB} in Lemma \ref{lem:union_confidence} and $\Tilde{E}_2$ corresponds to equation \eqref{eq:finite_Z_hoeffding}. By Lemma \ref{lem:confidence} and the Hoeffding-Azuma inequality, respectively, the following holds $\mathbb{P}(\Tilde{E}_1^C)\leq \delta/2$ and $\mathbb{P}(\Tilde{E}_2^C)\leq \delta/2$.
\end{proof}

\subsubsection{Infinite (large) number of contexts}\label{sec:infinite_context}

If the context space $\mathcal{Z}$ is large or infinite, we additionally assume that under the optimal feasible policy similar contexts lead to a similar performance \cite{sessa21}. 
In the design of a thermal controller, contexts could refer to weather conditions which forms an infinitely larges set. Then,  similar weather condition, e.g. the outdoor temperature,  result in similar control input parameters under the optimal feasible policy. 

\begin{ass}\label{ass:lipschitz}
The optimal feasible policy in hindsight $\overline{\pi}_i^*(\cdot)$ in \eqref{problem_context}  is $L_p$-Lipschitz, i.e.,
    \begin{align*}
        |\overline{\pi}_i^*(z)-\overline{\pi}_i^*(z')|\leq L_p\|z-z'\|_1,\quad \forall z,z'\in\mathcal{Z}.
    \end{align*}
Moreover, the reward function $r_i(\cdot)$ is $L_r$-Lipschitz with respect to the first component, i.e.,
    \begin{align*}
        |r_i(a_i,a_{-i},z)-r_i(a_i',a_{-i},z)|\leq L_r\|a_i-a_i'\|_1,
    \end{align*}
    for all $a_i,a_i'\in\mathcal{A}_i$, all $a_{-i}\in\mathcal{A}_{-i}$ and $z\in\mathcal{Z}$.
\end{ass} 
Furthermore, we assume $\mathcal{Z}\subseteq[0,1]^d$ to obtain a scale-free regret bound and for simplicity, we consider context-independent constraints $g_m:\mathcal{A}_i\rightarrow\mathbb{R}$, $m\in[M]$. Assumption \ref{ass:feasibility_context} then implies that $g_m(\overline{\pi}_i^*(z^t))\leq 0$ for all $m\in[M]$. 

Similarly to \cite{hazan07, sessa21}, the distribution $p^t(z^t)\in\Delta_K$ (line 6 of Algorithm \ref{alg:contextual_constrainedGPHedge}) is computed by building an $\epsilon$-net \cite{clakson05, krauthgamer2004} of the context space $\mathcal{Z}$ in a greedy fashion as new contexts
are revealed. Concretely, at each round $t$, either a new $L_1$-ball centered at $z^t$ is created, if $z^t$ is more than $\epsilon$ away from the closest ball, or $z^t$ is assigned to the closest $L_1$-ball. 
In the latter case, $p^t(z^t)$ is computed via the general AdaNormalHedge \cite{luo15} update rule using only past game data of those rounds $\tau<t$ for which $z^\tau$ belongs to the same $L_1$-ball as $z^t$. Under this computation rule of $p^t(z^t)$, summarized in Algorithm \ref{alg:strategy_infinite_Z}, c.z.AdaNormalGP achieves the following high-probability bounds.

\begin{algorithm*}[t]
\caption{Update rule for infinite context space $\mathcal{Z}$}\label{alg:strategy_infinite_Z}
\hspace*{\algorithmicindent} \textbf{Set} $\epsilon>0$, $\mathcal{C}=\{z^1\}$, $R^1_{a_i}(z^1)=0$, $C^1_{a_i}(z^1)=0$, and let $p^1(z^1)$ be the uniform distribution.
\begin{algorithmic}[1]
    \For {\texttt{$t=2,\ldots,T$}}
        \State \texttt{Set $\overline{z}^t=\arg\min_{z\in\mathcal{C}}\|z^t-z\|_1$.}
        \If{$\|z^t-\overline{z}^t\|_1>\epsilon$}
            \State \texttt{Add $z^t$ to $\mathcal{C}$, set $\overline{z}^t=z^t, R^1_{a_i}(\overline{z}^t)=0, C^1_{a_i}(\overline{z}^t)=0$, and $p^t(\overline{z}^t)$ as the uniform distribution.}
        \Else
            \State \texttt{Compute reward estimate $\hat{r}^t$ with $\hat{r}^t(a_i)=\min\{1,\text{UCB}_0^t(a_i,a_{-i}^t,z^t)\},\quad\forall a_i$.}
        \State \texttt{Compute for $\overline{z}^t$ and all $a_i\in\mathcal{A}_i$, $R^t_{a_i}(\overline{z}^t)$ and $C_{a_i}^t(\overline{z}^t)$ via eq. \eqref{alg:strategy_finite_Z_R} and  \eqref{alg:strategy_finite_Z_C} in Alg. \ref{alg:strategy_finite_Z}.}
        \State \texttt{Update $p_{a_i}^{t+1}(\overline{z}^t)$ via eq. \eqref{alg:strategy_finite_Z_update} in Alg. \ref{alg:strategy_finite_Z} and set $p^{t+1}(z^t)=p^{t+1}(\overline{z}^t)$.}
        \EndIf
    \EndFor
\end{algorithmic}
\end{algorithm*}

\begin{theorem}\label{thm:infinite_Z}
    Fix $\delta\in(0,1)$. Under Assumptions \ref{ass:feasibility_context}-\ref{ass:lipschitz}, if a player plays according to c.z.AdaNormalGP with $p^t(z^t)$ computed according to Algorithm \ref{alg:strategy_infinite_Z} and {\small$\beta_m^t=B_m+\sigma_m\sqrt{2(\gamma_m^{t-1}+1+\log(2(M+1)/\delta))}$} for all {\small$m\in\{0\}\cup[M]$}, then with probability at least $1-\delta$:
    {\small
    \begin{align*}
    &R^T=\mathcal{O}\big((L_r L_p)^{\frac{d}{d+2}} T^{\frac{d+1}{d+2}}\sqrt{\log(K)+\log(B)+\log(1+\log(K))}+\sqrt{T/2\log(2/\delta)}+2\beta_0^T\sqrt{T\gamma_0^T}\big)\\
    &\mathcal{V}_{m}^T=\mathcal{O}\big(\beta_m^T\sqrt{T\gamma_m^T}\big),\hspace{0,3cm}\forall m\in[M],
    \end{align*}
    }
    where {\small$B=1+\frac{3}{2}\frac{1}{K}\sum_{a_i=1}^K (1+\log(1+C_i^t(a_i)))\leq \frac{5}{2}+\frac{3}{2}\log(1+T)$}.
\end{theorem}

In repeated contextual games, c.GPMW \cite{sessa21} algorithm also considers infinite context spaces and achieves similar contextual regret bound for this setting, namely, {\small$\mathcal{O}((L_r L_p)^{\frac{d}{d+2}} T^{\frac{d+1}{d+2}}\sqrt{\log K}+\sqrt{T\log(2/\delta)}+\beta_0^T\sqrt{T\gamma_0^T})$}. In the following, we provide a proof for Theorem \ref{thm:infinite_Z}.



\begin{proof}
    
Lemmas \ref{lem:union_confidence}-\ref{lem:cum_violation_bound} in Appendix \ref{app:finite_Z} remain valid for infinite context space $\mathcal{Z}$ and context-independent constraints $g_m:\mathcal{A}_i\rightarrow \mathbb{R}$, where Assumption \ref{ass:feasibility_context} implies $g_m(\overline{\pi}_i^*(z^t))\leq 0$ for all $m\in[M]$. Thus, to prove Theorem \ref{thm:infinite_Z} it remains to show an upper bound on the regret. We leverage that $|\mathcal{C}|$ separate distributions for all $\epsilon$-close contexts $z\in\mathcal{Z}$ are maintained, where $|\mathcal{C}|$ is the total number of $L_1$-balls created up to round $T$.  Thus, similarly to the proof of Theorem \ref{thm:finite_Z}, we analyze the regret obtained by each $L_1$-ball separately. For each $L_1$-ball (fixed $z\in\mathcal{C}$), the proof further follows by decomposing the regret into the sum of two terms. The first term corresponds to the regret that a player incurs by considering the same distribution for $\epsilon$-close contexts rather than a separate distribution for each context $z\in\mathcal{Z}$. However, this term can be upper-bounded by making use of Lipschitzness (Assumption \ref{ass:lipschitz}). The second term can be bounded using the same proof techniques as in Theorem \ref{thm:finite_Z}. The theorem then follows from the same standard probability arguments as in the proof of Theorem \ref{thm:finite_Z}. 

\begin{lemma}\label{lem:regret_bound_infinite_Z}
Conditioned on equation \eqref{eq:simultaneous_UCB_LCB} in Lemma \ref{lem:union_confidence} holding true, the constrained contextual regret $R^T$ is upper-bounded with probability $1-\frac{\delta}{2}$ by:
{\small
\begin{align*}
    R^T=\mathcal{O}\bigg((L_r L_p)^{\frac{d}{d+2}} T^{\frac{d+1}{d+2}}\sqrt{\log(K)+\log(B)+\log(1+\log(K))}+\sqrt{T/2\log(2/\delta)}+2\beta_0^T\sqrt{T\gamma_0^T}\bigg).
\end{align*}}
\end{lemma}

\begin{proof}

Note that for infinite (large) context space $\mathcal{Z}$, c.z.AdaNormalGP builds an $\epsilon$-net of the context space by creating new $L_1$-balls in a greedy fashion to compute $p^t(z^t)$ (Algorithm \ref{alg:strategy_infinite_Z}). Thus, after $T$ game rounds, the $\epsilon$-net consists of the set of contexts $z\in\mathcal{Z}$, denoted by $\mathcal{C}$, that form the centers of the balls created so far. Moreover, the variable $\overline{z}^t$ indicates the center of the ball that context $z^t$ belongs to. Thus, the regret can be rewritten as:
    \begin{align*}
        R^T&=\sum_{z\in\mathcal{C}}\sum_{t:\overline{z}^t=z} r(\overline{\pi}_i^*(z^t),a_{-i}^t,z^t)-r(a_i^t,a_{-i}^t,z^t)\\
        &=\underbrace{\sum_{z\in\mathcal{C}}\sum_{t:\overline{z}^t=z} r(\overline{\pi}_i^*(z^t),a_{-i}^t,z^t)-  r(\overline{\pi}_i^*(z),a_{-i}^t,z^t) }_{L-R^T}+\underbrace{\sum_{z\in\mathcal{C}}\sum_{t:\overline{z}^t=z}r(\overline{\pi}_i^*(z),a_{-i}^t,z^t)-r(a_i^t,a_{-i}^t,z^t)}_{R-R^T},
    \end{align*}
    where in the last equality we add and subtract the term $r(\overline{\pi}_i^*(z),a_{-i}^t,z^t)$.

    The first term $L-R^T$ can be bounded by making use of Assumption \ref{ass:lipschitz}, namely, the $L_r$-Lipschitzness of $r(\cdot)$ in its first argument and $L_p$-Lipschitzness of the the optimal feasible policy $\overline{\pi}_i^*(\cdot)$:
    \begin{align}\label{eq:lipschitz_bound}
        L-R^T\leq \sum_{z\in\mathcal{C}}\sum_{t:\overline{z}^t=z} L_r\|\overline{\pi}_i^*(z^t)-\overline{\pi}_i^*(z)\|_1
        \leq \sum_{z\in\mathcal{C}}\sum_{t:\overline{z}^t=z} L_r L_p\|z^t-z\|_1\leq L_r L_p T\epsilon
    \end{align}
    where we used that $\|z^t-z\|_1\leq\epsilon$ by definition of the $\epsilon$-net (line $3$ and $6$ Algorithm \ref{alg:strategy_infinite_Z}).

    Next, we proceed with bounding $R-R^T$. 
The reward and thus $R-R^T$ can be further upper bounded by the same arguments and steps (see equation \eqref{eq:finite_Z,UCB,LCB} and \eqref{eq:finite_Z,beta,bound}) as in the proof of Lemma \ref{lem:regret_bound_finite_Z}:
    \begin{align}
        R-R^T&\leq \sum_{z\in\mathcal{C}}\sum_{t:\overline{z}^t=z}r(\overline{\pi}_i^*(z),a_{-i}^t,z^t)-r(a_i^t,a_{-i}^t,z^t)\nonumber\\
        &\leq \sum_{z\in\mathcal{C}}\sum_{t:\overline{z}^t=z} (\min\{1,\text{UCB}_0^t(\overline{\pi}_i^*(z),a_{-i}^t,z^t)\}- \min\{1,\text{UCB}_0^t(a_i^t,a_{-i}^t,z^t)\})+C_1\beta_0^T\sqrt{T\gamma_0^T}\label{eq:beta,bound},
    \end{align}
    where $C_1=\frac{8}{\log(1+\sigma^{-2})}$.

    The Hoeffding-Azuma inequality \cite[Lemma~A.7]{cesa-bianchi06} in \eqref{eq:finite_Z_hoeffding} remains applicable in the case where the context set $\mathcal{Z}$ is infinite. Thus, with probability at least $1-\delta/2$ we obtain the following upper bound on $R-R^T$:
    \begin{align}
        R-R^T&\leq  \sum_{z\in\mathcal{C}}\sum_{t:\overline{z}^t=z} \big(\min\{1,\text{UCB}_0^t(\overline{\pi}_i^*(z),a_{-i}^t,z^t)\}- \sum_{a_i\in\mathcal{A}_i} \overline{p}_{a_i}^t(z^t)\min\{1,\text{UCB}_0^t(a_i,a_{-i}^t,z^t)\}\big)\label{eq:AdaNormalHedgeRegret}\\
        &+\sqrt{T/2\log(2/\delta)}+C_1\beta_0^T\sqrt{T\gamma_0^T}\nonumber.
    \end{align}
    
    Define the function $f^t(\cdot)=\min\{1,\text{UCB}_0^t(\cdot,a_{-i}^t,z^t)\}$. Since we conditioned on equation \eqref{eq:simultaneous_UCB_LCB} in Lemma \ref{lem:union_confidence} holding true, $\text{UCB}_0^t(\cdot)\geq 0$ since $r(\cdot)\geq 0$ and thus $f^t(\cdot)\in[0,1]^K$. Similarly to equation \eqref{eq:finite_Z_normal_regret} and \eqref{eq:finite_Z_sleeping_regret}, for a fixed $z\in\mathcal{C}$, the first term in equation \eqref{eq:AdaNormalHedgeRegret} can be rewritten as:
    \begin{align}
    &\sum_{t:\overline{z}^t=z} \min\{1,\text{UCB}_0^t(\overline{\pi}_i^*(z),a_{-i}^t,z^t)\}- \sum_{a_i\in\mathcal{A}_i}\overline{p}_{a_i}^t(z^t)\min\{1,\text{UCB}_0^t(a_i,a_{-i}^t,z^t)\}\nonumber\\
        =& \sum_{t:\overline{z}^t=z} f^t(\overline{\pi}_i^*(z))-\sum_{a_i\in\mathcal{A}_i}\overline{p}_{a_i}^t(z^t)f^t(a_i^t)\label{eq:normal_regret}\\
        =& \sum_{t:\overline{z}^t=z} \mathbbm{1}_{\overline{\pi}_i^*(z),t}\bigg(f^t(\overline{\pi}_i^*(z))-\sum_{a_i\in\mathcal{A}_i}\overline{p}_{a_i}^t(z^t)f^t(a_i^t)\bigg),\label{eq:sleeping_regret}
    \end{align}
    where for all $a_i\in\mathcal{A}_i$ we define:
    \begin{align}\label{eq:indicator_sleeping}
        \begin{split}
            \mathbbm{1}_{a_i,t}:=\begin{cases}
			1, & \text{if $\text{LCB}_m^t(a_i)\leq 0$ for every $m\in[M],$}\\
            0, & \text{otherwise}.
		 \end{cases}
        \end{split}
    \end{align}
    Since we conditioned on equation \eqref{eq:simultaneous_UCB_LCB} in Lemma \ref{lem:union_confidence} holding true, $\mathbbm{1}_{\overline{\pi}_i^*(z),t}$ equals $1$ for all $t\geq 1$ since $z\in\{z^1,\cdots,z^T\}$ and thus $\text{LCB}_m^t(\overline{\pi}_i^*(z))\leq g_m(\overline{\pi}_i^*(z))\leq 0$. Therefore, equation \eqref{eq:normal_regret} and equation \eqref{eq:sleeping_regret} are equivalent. Observe that equation \eqref{eq:sleeping_regret} is precisely the regret with respect to $\overline{\pi}_i^*(z)$ which a player with reward functions $f_i^t(\cdot)\in[0, 1]$ incurs in a sleeping experts problem after $T_z=\sum_{t=1}^T 1_{\{\overline{z}^t=z
    \}}$ repetitions of the game. Here $T_z$ denotes the number of times the revealed context $z^t$ belonged to the ball centered at $z\in\mathcal{C}$. Note, furthermore, that for each context $z\in\mathcal{C}$ a distribution $p^t(z)$ is maintained which is updated whenever context $z^t$ belongs to the ball with center $z$. The update rule for each distribution (Line $7-8$ in Algorithm \ref{alg:strategy_infinite_Z}) corresponds exactly to that of general AdaNormalHedge proposed in \cite{luo14a}. Thus, at each round $t$, action $a_i^t$ is chosen according to general AdaNormalHedge with the set of sleeping experts $[\mathbbm{1}_{a_i,t}]_{a_i\in\mathcal{A}_i}$ defined as in equation \eqref{eq:indicator_sleeping} and which receives the full information feedback $\hat{r}^t=[f^t(a_i)]_{a_i\in\mathcal{A}_i}$.

    In summary, equation \eqref{eq:sleeping_regret} corresponds to the regret incurred by the general AdaNormalHedge algorithm with respect to expert $\overline{\pi}_i^*(z)\in\mathcal{A}_i$ after $T_z$ repetitions and can therefore be upper-bounded by Theorem \ref{thm:pseudo_ada} as follows:
    \begin{align}\label{eq:sleeping_regret_bound}
        \sum_{t:\overline{z}^t=z} \mathbbm{1}_{\overline{\pi}_i^*(z),t}(f^t(\overline{\pi}_i^*(z))-f^t(a_i^t))\leq\sqrt{3T_z(\log(K)+\log(B)+\log(1+\log(K)))},
    \end{align}
    where $B=\leq \frac{5}{2}+\frac{3}{2}\log(1+T_z)$. 
    
    Now summing over all the contexts $z\in\mathcal{C}$ and using the Cauchy-Schwarz inequality, we can further upper-bound equation \eqref{eq:AdaNormalHedgeRegret} with probability at least $1-\delta/2$:
    \begin{align}
        R-R^T&\leq\sum_{z\in\mathcal{C}}\sum_{t:\overline{z}^t=z} \min\{1,\text{UCB}_0^t(\overline{\pi}_i^*(z),a_{-i}^t,z^t)\}- \sum_{a_i\in\mathcal{A}_i}\overline{p}_{a_i}^t(z^t)\min\{1,\text{UCB}_0^t(a_i,a_{-i}^t,z^t)\}\nonumber\\
        &+\sqrt{T/2\log(2/\delta)}+C_1\beta_0^T\sqrt{T\gamma_0^T}\nonumber\\
    &\leq \sum_{z\in\mathcal{C}}\sqrt{3T_z(\log(K)+\log(B)+\log(1+\log(K)))}+\sqrt{T/2\log(2/\delta)}+C_1\beta_0^T\sqrt{T\gamma_0^T}\nonumber\\
    &\leq \sqrt{|\mathcal{C}|\sum_{z\in\mathcal{C}}3T_z(\log(K)+\log(B)+\log(1+\log(K)))}+\sqrt{T/2\log(2/\delta)}+C_1\beta_0^T\sqrt{T\gamma_0^T}\nonumber\\
    &=\sqrt{|\mathcal{C}|3T(\log(K)+\log(B)+\log(1+\log(K)))}+\sqrt{T/2\log(2/\delta)}+C_1\beta_0^T\sqrt{T\gamma_0^T}\nonumber\\
    &\leq \epsilon^{-d/2}\sqrt{3T(\log(K)+\log(B)+\log(1+\log(K)))}+\sqrt{T/2\log(2/\delta)}+C_1\beta_0^T\sqrt{T\gamma_0^T}.\label{eq:R-R_i^T}
\end{align}
In the last inequality we used that $\mathcal{C}\leq (1/\epsilon)^d$, because the contexts space $\mathcal{Z}\subseteq[0,1]^d$ can be covered by at most $(1/\epsilon)^d$ balls of radius $\epsilon$ such that the distance between their centers is at least $\epsilon$ \cite{clakson05}.

Combining the bounds of $L-R^T$ in \eqref{eq:lipschitz_bound} and $R-R^T$ in \eqref{eq:R-R_i^T}, with probability at least $1-\delta/2$ the regret is bounded by:
\begin{align*}
    R^T&= L-R^T+R-R^T\\
    &\leq L_r L_p T\epsilon+\epsilon^{-d/2}\sqrt{3T(\log(K)+\log(B)+\log(1+\log(K)))}\\
&+\sqrt{T/2\log(2/\delta)}+C_1\beta_0^T\sqrt{T\gamma_0^T}\\
    &=(L_r L_p)^{\frac{d}{d+2}} T^{\frac{d+1}{d+2}}(1+\sqrt{3(\log(K)+\log(B)+\log(1+\log(K)))})\\
    &+\sqrt{T/2\log(2/\delta)}+C_1\beta_0^T\sqrt{T\gamma_0^T},
\end{align*}
where in the last equality $\epsilon$ is set as $(L_r L_p)^{-\frac{2}{d+2}}T^{-\frac{1}{d+2}}$.
\end{proof}
\end{proof}

\subsection{Supplementary Material for Remark \ref{rem:general_bounds}}\label{app:general_regret_bounds}

In this section, we show that we can obtain high probability bounds on the regret and the cumulative constraint violations for Algorithm \ref{alg:contextual_constrainedGPHedge} presented in Section \ref{sec:c.z.AdaNormalGP}, when other computation rules of $p^t(z^t)$ (line $6$ of Algorithm \ref{alg:contextual_constrainedGPHedge}) than the update rule of AdaNormalHedge \cite{luo15} are used. Recall that once $p^t(z^t)$ is specified, Algorithm \ref{alg:contextual_constrainedGPHedge} is well-defined.

\subsubsection{Expert algorithm bounds for finite context space}\label{app:general_regret_bounds_finite_Z}

In this subsection, we present general regret bounds that depend on the specified computation rule $p^t(z^t)$ for the case when the context space $\mathcal{Z}$ is finite. As in Section \ref{sec:finite_context}, the distribution $p^t(z^t)$ is computed by maintaining a distribution for each context $z\in\mathcal{Z}$. For a fixed context $z\in\mathcal{Z}$, however, any update rule that corresponds to an update rule of a sleeping expert problem with regret bounds can be used to obtain a high probability regret bound. This follows from the connection between playing a repeated game with unknown constraints and playing a sleeping expert problem. In particular, since any regular expert algorithm such as Hedge \cite{freund95} can be reformulated as a sleeping expert algorithm \cite{freund1997} with the same regret bound, for a fixed context $z\in\mathcal{Z}$, any update rule that corresponds to the reformulated update rule of an expert problem can be used to obtain a high probability regret bound. This bound then depends on the regret obtained by the expert algorithm. Algorithm \ref{alg:c.z.no-regret_finite_Z} provides a computation rule for $p^t(z^t)$ which is based on the reformulation of the update rule of any no-regret expert algorithm $\mathcal{E}$ to obtain a sleeping expert algorithm. The bounds on the cumulative constraint violations obtained in Theorem \ref{thm:finite_Z} are independent of the update rule and therefore remain valid. 

\begin{theorem}\label{thm:general_algo_finite_Z}
    Fix $\delta\in(0,1)$. Under the same assumptions as in Theorem \ref{thm:finite_Z}, if a player plays according to Algorithm \ref{alg:contextual_constrainedGPHedge} with {\small$p^t(z^t)$} computed according to Algorithm \ref{alg:c.z.no-regret_finite_Z} then with probability at least $1-\delta$:
    \begin{align*}
    &R^T=\mathcal{O}\bigg(\sqrt{|\mathcal{Z}|\sum_{z\in\mathcal{Z}}(R^{T_z}(\mathcal{E}))^2}+\sqrt{T\log(2/\delta)}+\beta_0^T\sqrt{T\gamma_0^T}\bigg)\\
    &\mathcal{V}_{m}^T=\mathcal{O}\bigg(\beta_m^T\sqrt{T\gamma_m^T}\bigg),\hspace{0,3cm}\forall m\in[M],
    \end{align*}
    where $R^{T_z}(\mathcal{E})$ is the regret obtained by an expert algorithm $\mathcal{E}$ after $T_z$ rounds and {\small$T_z=\sum_{t=1}^T\mathbbm{1}_{\{z^t=z\}}$}. We denote Algorithm \ref{alg:contextual_constrainedGPHedge} with {\small$p^t(z^t)$} computed according to Algorithm \ref{alg:c.z.no-regret_finite_Z} as c.z.$\mathcal{E}$GP.
\end{theorem}

For example, when using the update rule of the well-known no-regret algorithm Hedge \cite{freund95} to compute $p^t(z^t)$ in Algorithm \ref{alg:c.z.no-regret_finite_Z}, we obtain the following bound on the regret.
\begin{cor}\label{cor:hedge_finite_Z}
    If the expert algorithm $\mathcal{E}$ corresponds to the Hedge algorithm with step size {\small$\eta^t=2\sqrt{\log K/\sum_{\tau=1}^t\mathbbm{1}_{\{z^\tau=z^t\}}}$} and prediction rule:

\begin{align*}
    p_{a_i}^{t+1}(z)=p_{a_i}^t(z)\exp(\eta^t\text{UCB}_0^t(a_i,a_{-i}^t,z^t)\mathbbm{1}_{\{z^t=z\}}, \quad\forall a_i\in\mathcal{A}_i,\ z\in\mathcal{Z},
\end{align*}

then playing according to the so-called c.z.HedgeGP algorithm\footnote{Here, we consider the case where $p^t(z^t)$ is computed according to Algorithm \ref{alg:c.z.no-regret_finite_Z}.} yields the following high-probability regret bound:
\begin{align*}
    R^T=\mathcal{O}(\sqrt{|\mathcal{Z}|T \log K}+\sqrt{T\log(2/\delta)}+\beta_0^T\sqrt{T\gamma_0^T}).
\end{align*} 

\end{cor}

\begin{proof}
    The proof follows from Theorem \ref{thm:general_algo_finite_Z} and \cite[Proposition~1]{mourtada2019} which states that with $\eta^t$ defined as in Corollary \ref{cor:hedge_finite_Z}, the regret of Hedge which receives full-information feedback $\textbf{r}^t\in[0,1]^K$ is bounded as follows $\sum_{t=1}^T \big(r^t(a_i^*)-\sum_{a_i\in\mathcal{A}_i}p_{a_i}^tr^t(a_i^t)\big)\leq\sqrt{T\log K}$.
\end{proof}

\begin{algorithm}[t]
\caption{Update rule for finite context space $\mathcal{Z}$}\label{alg:c.z.no-regret_finite_Z}
\hspace*{\algorithmicindent} \textbf{Input:} Full-information \textbf{no-regret algorithm $\mathcal{E}$}. Let $p^1(z)$ be the uniform distribution $\forall z\in\mathcal{Z}$.
\begin{algorithmic}[1]
    \For {\texttt{$t=2,\ldots,T$}}
        \State \texttt{Compute reward estimate $\hat{\textbf{r}}^t$ with
        \begin{align}
            \hat{r}^t(a_i)=\begin{cases}
        \min\{1,\text{UCB}_0^t(a_i,a_{-i}^t,z^t)\}, & \text{if $\text{LCB}_m^t(a_i,z^t)\leq 0$, $\forall m\in[M],$}\\
        \langle\overline{p}^t(z^t),\hat{r}^t\rangle, & \text{otherwise}.
        \end{cases}
        \end{align}}
        \State \texttt{Pass $\hat{\textbf{r}}^t\in\mathbb{R}^K$ and $p^t(z^t)\in\mathbb{R}^K$ to algorithm $\mathcal{E}$.}
        \State \texttt{Let $p^{t+1}(z^t)$ be the prediction of algorithm $\mathcal{E}$.}
    \EndFor
\end{algorithmic}
\end{algorithm}

\begin{proof}(of Theorem \ref{thm:general_algo_finite_Z}) Lemma \ref{lem:union_confidence}-\ref{lem:cum_violation_bound} in Appendix \ref{app:finite_Z} remain valid since the same assumptions as in Theorem \ref{thm:finite_Z} are made. Thus, to prove Theorem \ref{thm:general_algo_finite_Z} it is left to show that conditioned on equation \eqref{eq:simultaneous_UCB_LCB} in Lemma \ref{lem:union_confidence} holding true,  the regret $R^T$ is upper-bounded with probability at least $1-\frac{\delta}{2}$ by:
    \begin{align}\label{eq:general_regret}
        R^T\leq\sqrt{|\mathcal{Z}|\sum_{z\in\mathcal{Z}}(R^{T_z}(\mathcal{E}))^2}+\sqrt{T/2\log(2/\delta)}+C_1\beta_0^T\sqrt{T\gamma_0^T}.
    \end{align}
The theorem then follows from the same standard probability arguments as in the proof of Theorem \ref{thm:finite_Z}.

The main difference between Lemma 
\ref{lem:regret_bound_finite_Z} used to prove Theorem \ref{thm:finite_Z} and the statement in equation \eqref{eq:general_regret} is that the update rule $p^t(z^t)$ is computed according to Algorithm \ref{alg:strategy_finite_Z} for the former and according to Algorithm \ref{alg:c.z.no-regret_finite_Z} for the latter. Thus, to show the statement in equation \eqref{eq:general_regret}, any steps and arguments from Lemma \ref{lem:regret_bound_finite_Z} that are independent of the computation rule of $p^t(z^t)$ can be repeated.

Since c.z.$\mathcal{E}$GP (Algorithm \ref{alg:contextual_constrainedGPHedge} with $p^t(z^t)$ computed according to Algorithm \ref{alg:c.z.no-regret_finite_Z}) maintains a separate distribution $p^t(z)$ for each context $z\in\mathcal{Z}$ the regret notion can be rewritten as:
\begin{align*}
        R^T&=\sum_{z\in\mathcal{Z}}\sum_{t:z^t=z} r(\overline{\pi}_i^*(z^t),a_{-i}^t,z^t)-r(a_i^t,a_{-i}^t,z^t).
\end{align*}

The regret $R^T$ can be further upper-bounded by the same arguments and steps as in equations \eqref{eq:finite_Z,UCB,LCB}, \eqref{eq:finite_Z,beta,bound} and \eqref{eq:finite_Z_hoeffding}. Then, with probability at least {\small$1-\frac{\delta}{2}$} the following upper bound holds:
    \begin{align}
        R^T&\leq \sum_{z\in\mathcal{Z}}\sum_{t:z^t=z}r(\overline{\pi}_i^*(z^t),a_{-i}^t,z^t)-r(a_i^t,a_{-i}^t,z^t)\nonumber\\
        &\leq \sum_{z\in\mathcal{Z}}\sum_{t:z^t=z} (\min\{1,\text{UCB}_0^t(\overline{\pi}_i^*(z^t),a_{-i}^t,z^t)\}- \min\{1,\text{UCB}_0^t(a_i^t,a_{-i}^t,z^t)\})+C_1\beta_0^T\sqrt{T\gamma_0^T}\nonumber\\
        &\leq  \sum_{z\in\mathcal{Z}}\sum_{t:z^t=z} \big(\min\{1,\text{UCB}_0^t(\overline{\pi}_i^*(z^t),a_{-i}^t,z^t)\}- \sum_{a_i\in\mathcal{A}_i} \overline{p}_{a_i}^t(z^t)\min\{1,\text{UCB}_0^t(a_i,a_{-i}^t,z^t)\}\big)\label{eq:finite_Z_c.z.AGP}\\
        &+\sqrt{T/2\log(2/\delta)}+C_1\beta_0^T\sqrt{T\gamma_0^T}\nonumber.
    \end{align}
    where $C_1=\frac{8}{\log(1+\sigma^{-2})}$.

    Recall the function $\hat{\textbf{r}}^t$ (line $2$ of Algorithm \ref{alg:c.z.no-regret_finite_Z}) defined as:
     \begin{align}
            \hat{r}^t(a_i)=\begin{cases}
        \min\{1,\text{UCB}_0^t(a_i,a_{-i}^t,z^t)\}, & \text{if $\text{LCB}_m^t(a_i,z^t)\leq 0$, $\forall m\in[M],$}\\
        \langle\overline{p}^t(z^t),\hat{r}^t\rangle, & \text{otherwise}.
        \end{cases}
        \end{align}
    For a fixed $z\in\mathcal{Z}$, the first term in equation \eqref{eq:finite_Z_c.z.AGP} can be rewritten as:
    \begin{align}
    &\sum_{t:z^t=z} \min\{1,\text{UCB}_0^t(\overline{\pi}_i^*(z^t),a_{-i}^t,z^t)\}- \sum_{a_i\in\mathcal{A}_i}\overline{p}_{a_i}^t(z^t)\min\{1,\text{UCB}_0^t(a_i,a_{-i}^t,z^t)\}\label{eq:def_reformulated_r}\\
        =& \sum_{t:z^t=z} \hat{r}^t(\overline{\pi}_i^*(z^t))-\sum_{a_i\in\mathcal{A}_i}\overline{p}_{a_i}^t(z^t)\hat{r}^t(a_i)\label{eq:general_finite_Z_normal_regret}\\
        =& \sum_{t:z^t=z} \mathbbm{1}_{\overline{\pi}_i^*(z^t),t}\bigg(\hat{r}^t(\overline{\pi}_i^*(z^t))-\sum_{a_i\in\mathcal{A}_i}\overline{p}_{a_i}^t(z^t)\hat{r}^t(a_i)\bigg),\label{eq:general_finite_Z_sleeping_regret}\\
        =&\sum_{t:z^t=z} \bigg(\hat{r}^t(\overline{\pi}_i^*(z^t))-\sum_{a_i\in\mathcal{A}_i}p_{a_i}^t(z^t)\hat{r}^t(a_i)\bigg),\label{eq:conversion_sleeping_normal_regret}
    \end{align}
    where for all $a_i\in\mathcal{A}_i$ the indicator function $\mathbbm{1}_{a_i,t}$ is defined as in equation \eqref{eq:finite_Z_indicator_sleeping}

Equation \eqref{eq:general_finite_Z_normal_regret} follows by definition of $\hat{\textbf{r}}$ (line $2$ Algorithm \ref{alg:c.z.no-regret_finite_Z}). Concretely, since we conditioned on equation \eqref{eq:simultaneous_UCB_LCB} in Lemma \ref{lem:union_confidence} holding true, $\text{LCB}_m^t(\overline{\pi}_i^*(z^t),z^t)\leq 0$ for all $m\in[M]$ and thus {\small$\hat{r}^t(\overline{\pi}_i^*(z^t))=\min\{1,\text{UCB}_0^t(\overline{\pi}_i^*(z^t),a_{-i}^t,z^t)\}$}. Furthermore, note that $\overline{p}_{a_i}^t(z^t)=0$ if there exists an $m\in[M]$ such that $\text{LCB}_m^t(a_i,z^t)>0$ and in this case {\small$\overline{p}_{a_i}^t(z^t)\min\{1,\text{UCB}_0^t(a_i,a_{-i}^t,z^t)\}=\overline{p}_{a_i}^t(z^t)\hat{r}^t(a_i)=0$}. Since we conditioned on equation \eqref{eq:simultaneous_UCB_LCB} in Lemma \ref{lem:union_confidence} holding true, $\mathbbm{1}_{\overline{\pi}_i^*(z^t),t}$ equals $1$ for all $t\geq 1$ and therefore equation \eqref{eq:general_finite_Z_normal_regret} and equation \eqref{eq:general_finite_Z_sleeping_regret} are equivalent. Equation \eqref{eq:conversion_sleeping_normal_regret} follows from the fact that any expert problem can be reformulated as a sleeping expert problem and vice versa \cite{freund1997}. Specifically, if we are given a no-regret expert algorithm $\mathcal{E}$ with prediction $p^t(z^t)$ on round $t$ then the prediction of the corresponding sleeping expert algorithm is given by $\overline{p}_{a_i}^t(z^t)\propto \mathbbm{1}_{a_i,t} p_{a_i}^t(z^t)$ on round $t$. Furthermore, the reward vector $\hat{\textbf{r}}^t$ which is given to $\mathcal{E}$ as full-information feedback is designed such that:
\begin{align*}   
    \sum_{a_i:\mathbbm{1}_{a_i,t}=1}p_{a_i}^t(z^t)\hat{r}^t(a_i)+\sum_{a_i:\mathbbm{1}_{a_i,t}=0}p_{a_i}^t(z^t)\hat{r}^t(a_i)=\sum_{a_i:\mathbbm{1}_{a_i,t}=1}\overline{p}_{a_i}^t(z^t)\hat{r}^t(a_i).
\end{align*}
thus, by design the equivalence between equation \eqref{eq:general_finite_Z_sleeping_regret} and equation \eqref{eq:conversion_sleeping_normal_regret} follows. Note furthermore, that since we conditioned on equation \eqref{eq:simultaneous_UCB_LCB} in Lemma \ref{lem:union_confidence} holding true, $\text{UCB}_0^t(\cdot)\geq 0$ since $r(\cdot)\geq 0$ and thus $\hat{\textbf{r}}^t\in[0,1]^K$. Now observe that equation \eqref{eq:conversion_sleeping_normal_regret} is precisely the regret which a player with reward function $\hat{r}^t(\cdot)\in[0, 1]$ incurs in an expert problem after $T_z=\sum_{t=1}^T 1_{\{z^t=z\}}$ repetitions of the game. Here $T_z$ denotes the number of times context $z^t$ is revealed. The prediction $p^t(z^t)$ for each context $z\in\mathcal{Z}$ (Line $4$ in Algorithm \ref{alg:c.z.no-regret_finite_Z}) is given by the no-regret algorithm $\mathcal{E}$ which receives full-information feedback $\hat{\textbf{r}}^t=[\hat{r}^t(a_i)]_{a_i\in\mathcal{A}_i}$. Thus, equation \eqref{eq:conversion_sleeping_normal_regret} corresponds to the regret incurred by the algorithm $\mathcal{E}$ after $T_z$ repetitions, i.e., 
\begin{align}
    \sum_{t:z^t=z} \bigg(\hat{r}^t(\overline{\pi}_i^*(z^t))-\sum_{a_i\in\mathcal{A}_i}p_{a_i}^t(z^t)\hat{r}^t(a_i^t)\bigg)\leq R^{T_z}(\mathcal{E}).
\end{align}
    
    Now summing over all contexts $z\in\mathcal{Z}$ and using the Cauchy-Schwarz inequality, we can further upper-bound equation \eqref{eq:finite_Z_c.z.AGP} with probability at least $1-\delta/2$:
    \begin{align}
        R^T\leq&\sum_{z\in\mathcal{Z}}\sum_{t:z^t=z} \min\{1,\text{UCB}_0^t(\overline{\pi}_i^*(z^t),a_{-i}^t,z^t)\}- \sum_{a_i\in\mathcal{A}_i}\overline{p}_{a_i}^t(z^t)\min\{1,\text{UCB}_0^t(a_i,a_{-i}^t,z^t)\}\nonumber\\
        &+\sqrt{T/2\log(2/\delta)}+C_1\beta_0^T\sqrt{T\gamma_0^T}\nonumber\\
    &\leq \sum_{z\in\mathcal{Z}}R^{T_z}(\mathcal{E})+\sqrt{T/2\log(2/\delta)}+C_1\beta_0^T\sqrt{T\gamma_0^T}\nonumber\\
    &\leq \sqrt{|\mathcal{Z}|\sum_{z\in\mathcal{Z}}(R^{T_z}(\mathcal{E}))^2}+\sqrt{T/2\log(2/\delta)}+C_1\beta_0^T\sqrt{T\gamma_0^T}\nonumber
\end{align}
\end{proof}

\subsubsection{Expert algorithm bounds for infinite context space}\label{app:general_regret_bounds_infinite_Z}

In this subsection, we present general regret bounds that depend on the specified computation rule $p^t(z^t)$ for the case when the context space $\mathcal{Z}$ is infinite. As in Appendix \ref{sec:infinite_context}, the distribution $p^t(z^t)$ is computed by maintaining $|\mathcal{C}|$ separate distributions for all $\epsilon$-close contexts $z\in\mathcal{Z}$, where $|\mathcal{C}|$ is the total number of $L_1$-balls created up to round $T$. For each $L_1$-ball (fixed $z\in\mathcal{C}$), any update rule that corresponds to the reformulated update rule of an expert problem can be used to obtain a high probability regret bound  which then depends on the regret obtained by the expert algorithm.

\begin{theorem}\label{thm:general_algo_infinite_Z}
    Fix $\delta\in(0,1)$. Under the same assumptions as in Theorem \ref{thm:infinite_Z}, if a player plays according to Algorithm \ref{alg:contextual_constrainedGPHedge} with {\small$p^t(z^t)$} computed according to Algorithm \ref{alg:c.z.no-regret_infinite_Z}, then with probability at least $1-\delta$
    \begin{align*}
    &R^T=\mathcal{O}\bigg((L_r L_p)^{\frac{d}{d+2}} T^{\frac{d}{d+2}} \bigg(\sum_{z\in\mathcal{C}} (R^{T_z}(\mathcal{E}))^2\bigg)^{\frac{1}{d+2}}+\sqrt{T\log(2/\delta)}+\beta_0^T\sqrt{T\gamma_0^T}\bigg)\\
    &\mathcal{V}_{m}^T=\mathcal{O}\bigg(\beta_m^T\sqrt{T\gamma_m^T}\bigg),\hspace{0,3cm}\forall m\in[M],
    \end{align*}
    where $R^{T_z}(\mathcal{E})$ is the regret obtained by algorithm $\mathcal{E}$ after $T_z$ rounds and {\small$T_z=\sum_{t=1}^T\mathbbm{1}_{\{\overline{z}^t=z\}}$} is the number of times the revealed context $z^t$ belonged to a ball centered at $z\in\mathcal{Z}$.
\end{theorem}

For example, when using the update rule of the well-known no-regret algorithm Hedge \cite{freund95} to compute $p^t(z^t)$ in Algorithm \ref{alg:c.z.no-regret_infinite_Z}, we obtain the following bound on the regret.

\begin{cor}
    If algorithm $\mathcal{E}$ corresponds to the Hedge algorithm with step size {\small$\eta^t=2\sqrt{\log K/\sum_{\tau=1}^t\mathbbm{1}_{\{\overline{z}^t=\overline{z}^\tau\}}}$} and prediction rule:

\begin{align*}
    p_{a_i}^{t+1}(z)=p_{a_i}^t(z)\exp(\eta^t\text{UCB}_0^t(a_i,a_{-i}^t,z^t)\mathbbm{1}_{\{\overline{z}^t=z\}}, \quad\forall a_i\in\mathcal{A}_i,\ z\in\mathcal{Z},
\end{align*}

then playing according to the so-called c.z.HedgeGP algorithm\footnote{Here, we consider the case where $p^t(z^t)$ is computed according to Algorithm \ref{alg:c.z.no-regret_infinite_Z}.} yields the following high-probability regret bound:
\begin{align*}
    R^T=\mathcal{O}((L_r L_p)^{\frac{d}{d+2}} T^{\frac{d+1}{d+2}}\log(K)^{\frac{1}{d+2}}+\sqrt{T\log(2/\delta)}+\beta_0^T\sqrt{T\gamma_0^T}).
\end{align*} 

\end{cor}

\begin{proof}
    The proof follows from Theorem \ref{thm:general_algo_infinite_Z} and \cite[Proposition~1]{mourtada2019}.
\end{proof}

\begin{algorithm}[t]
\caption{Update rule for infinite (large) context space $\mathcal{Z}$}\label{alg:c.z.no-regret_infinite_Z}
\hspace*{\algorithmicindent} \textbf{Input:} Full-information \textbf{no-regret algorithm $\mathcal{E}$}.\\
\hspace*{\algorithmicindent} \textbf{Set:} $\epsilon>0$ and $\mathcal{C}=\{z^1\}$ and let $p^1(z^1)$ be the uniform distribution.
\begin{algorithmic}[1]
    \For {\texttt{$t=2,\ldots,T$}}
        \State \texttt{Set $\overline{z}^t=\arg\min_{z\in\mathcal{C}}\|z^t-z\|_1$.}
        \If{$\|z^t-\overline{z}^t\|_1>\epsilon$}
            \State \texttt{Add $z^t$ to $\mathcal{C}$, set $\overline{z}^t=z^t$ and let $p^t(\overline{z}^t)$ be\\ \hspace*{\algorithmicindent}
         \ \ the uniform distribution.}
        \Else
            \State \texttt{Compute reward estimate $\hat{\textbf{r}}^t$ with
        \begin{align}
            \hat{r}^t(a_i)=\begin{cases}
            \min\{1,\text{UCB}_0^t(a_i,a_{-i}^t,z^t)\}, & \text{if $\text{LCB}_m^t(a_i)\leq 0$, $\forall m\in[M],$}\\
            \langle\overline{p}^t(z^t),\hat{r}^t\rangle, & \text{otherwise}.
            \end{cases}
            \end{align}}
        \State \texttt{Pass $\hat{\textbf{r}}^t\in\mathbb{R}^K$ and $p^t(\overline{z}^t)\in\mathbb{R}^K$ to algorithm $\mathcal{E}$.}
        \State \texttt{Let $p^{t+1}(\overline{z}^t)$ be the prediction of $\mathcal{E}$ and set $p^{t+1}(z^t)=p^{t+1}(\overline{z}^t)$.}
        \EndIf
    \EndFor
\end{algorithmic}
\end{algorithm}

\begin{proof} (of Theorem \ref{thm:general_algo_infinite_Z}) Lemma \ref{lem:union_confidence}-\ref{lem:cum_violation_bound} in Appendix \ref{app:finite_Z} remain valid for infinite context space $\mathcal{Z}$ and context-independent constraints $g_m:\mathcal{A}_i\rightarrow \mathbb{R}$, where Assumption \ref{ass:feasibility_context} implies $g_m(\overline{\pi}_i^*(z^t))\leq 0$ for all $m\in[M]$. Thus, to prove Theorem \ref{thm:general_algo_infinite_Z} it is left to show that conditioned on equation \eqref{eq:simultaneous_UCB_LCB} in Lemma \ref{lem:union_confidence} holding true, the regret is upper-bounded with probability at least $1-\frac{\delta}{2}$ by:

\begin{align}\label{eq:general_regret_infinite_Z}
    R^T\leq(L_r L_p)^{\frac{d}{d+2}} T^{\frac{d}{d+2}} \bigg(\sum_{z\in\mathcal{C}} (R^{T_z}(\mathcal{E}))^2\bigg)^{\frac{1}{d+2}}+\sqrt{T/2\log(2/\delta)}+C_1\beta_0^T\sqrt{T\gamma_0^T}.
\end{align}
The theorem then follows from the same standard probability arguments as in the proof of Theorem \ref{thm:finite_Z}.

The main difference between Lemma 
\ref{lem:regret_bound_infinite_Z} and the statement in equation \eqref{eq:general_regret_infinite_Z} is that the update rule $p^t(z^t)$ is computed according to Algorithm \ref{alg:strategy_finite_Z} for the former and according to Algorithm \ref{alg:c.z.no-regret_finite_Z} for the latter. Thus, to show the statement in equation \eqref{eq:general_regret_infinite_Z}, any steps and arguments from Lemma \ref{lem:regret_bound_infinite_Z} that are independent of the computation rule of $p^t(z^t)$ can be repeated.

Analogously to c.z.AdaNormalGP (Algorithm \ref{alg:contextual_constrainedGPHedge} with $p^t(z^t)$ computed by Algorithm \ref{alg:strategy_infinite_Z}), c.z.$\mathcal{E}$GP (Algorithm \ref{alg:contextual_constrainedGPHedge} with $p^t(z^t)$ computed by Algorithm \ref{alg:c.z.no-regret_infinite_Z}) builds an $\epsilon$-net of the context space by creating new $L_1$-balls in a greedy fashion. Here, the variable $\overline{z}^t$ indicates the ball that context $z^t$ belongs to. Then, the regret can be rewritten as:
    \begin{align*}
        R^T&=\sum_{z\in\mathcal{C}}\sum_{t:\overline{z}^t=z} r(\overline{\pi}_i^*(z^t),a_{-i}^t,z^t)-r(a_i^t,a_{-i}^t,z^t)\\
        &=\underbrace{\sum_{z\in\mathcal{C}}\sum_{t:\overline{z}^t=z} r(\overline{\pi}_i^*(z^t),a_{-i}^t,z^t)-  r(\overline{\pi}_i^*(z),a_{-i}^t,z^t) }_{L-R^T}+\underbrace{\sum_{z\in\mathcal{C}}\sum_{t:\overline{z}^t=z}r(\overline{\pi}_i^*(z),a_{-i}^t,z^t)-r(a_i^t,a_{-i}^t,z^t)}_{R-R^T},
    \end{align*}
    where in the last equality we add and subtract the term $r(\overline{\pi}_i^*(z),a_{-i}^t,z^t)$.

    The first term $L-R^T$ can be bounded by making use of Assumption \ref{ass:lipschitz} and then applying the same arguments as in the proof of Lemma \ref{lem:regret_bound_infinite_Z}:
    \begin{align}\label{eq:general_lipschitz_bound}
        L-R^T\leq L_r L_p T\epsilon.
    \end{align}

    Next, we proceed with bounding $R-R^T$. 
    It can be further upper-bounded by the same arguments and steps as in equations \eqref{eq:finite_Z,UCB,LCB}, \eqref{eq:finite_Z,beta,bound} and \eqref{eq:finite_Z_hoeffding}. Then, with probability at least {\small$1-\frac{\delta}{2}$} the following upper bound holds:
    \begin{align}
        R-R^T&\leq \sum_{z\in\mathcal{C}}\sum_{t:\overline{z}^t=z}r(\overline{\pi}_i^*(z),a_{-i}^t,z^t)-r(a_i^t,a_{-i}^t,z^t)\nonumber\\
        &\leq \sum_{z\in\mathcal{C}}\sum_{t:\overline{z}^t=z} (\min\{1,\text{UCB}_0^t(\overline{\pi}_i^*(z),a_{-i}^t,z^t)\}- \min\{1,\text{UCB}_0^t(a_i^t,a_{-i}^t,z^t)\})\nonumber\\
        &+C_1\beta_0^T\sqrt{T\gamma_0^T}\nonumber\\
        &\sum_{z\in\mathcal{C}}\sum_{t:\overline{z}^t=z} \big(\min\{1,\text{UCB}_0^t(\overline{\pi}_i^*(z),a_{-i}^t,z^t)\}- \sum_{a_i\in\mathcal{A}_i} \overline{p}_{a_i}^t(z^t)\min\{1,\text{UCB}_0^t(a_i,a_{-i}^t,z^t)\}\big)\label{eq:infinite_Z_c.z.AGP}\\
        &+\sqrt{T/2\log(2/\delta)}+C_1\beta_0^T\sqrt{T\gamma_0^T}\nonumber.
    \end{align}
    where $C_1=\frac{8}{\log(1+\sigma^{-2})}$.
    
    Recall the function $\hat{\textbf{r}}^t$ (line $7$ of Algorithm \ref{alg:c.z.no-regret_infinite_Z}) defined as:
     \begin{align}
            \hat{r}^t(a_i)=\begin{cases}
        \min\{1,\text{UCB}_0^t(a_i,a_{-i}^t,z^t)\}, & \text{if $\text{LCB}_m^t(a_i)\leq 0$, $\forall m\in[M],$}\\
        \langle\overline{p}^t(z^t),\hat{r}^t\rangle, & \text{otherwise}.
        \end{cases}
        \end{align}
    For a fixed $z\in\mathcal{C}$, the first term in equation \eqref{eq:infinite_Z_c.z.AGP} can be rewritten as:
    \begin{align}
    &\sum_{t:z^t=z} \min\{1,\text{UCB}_0^t(\overline{\pi}_i^*(z),a_{-i}^t,z^t)\}- \sum_{a_i\in\mathcal{A}_i}\overline{p}_{a_i}^t(z^t)\min\{1,\text{UCB}_0^t(a_i,a_{-i}^t,z^t)\}\nonumber\\
        =& \sum_{t:z^t=z} \hat{r}^t(\overline{\pi}_i^*(z^t))-\sum_{a_i\in\mathcal{A}_i}\overline{p}_{a_i}^t(z)\hat{r}^t(a_i^t)\label{eq:general_infinite_Z_normal_regret}\\
        =& \sum_{t:z^t=z} \mathbbm{1}_{\overline{\pi}_i^*(z),t}\bigg(\hat{r}^t(\overline{\pi}_i^*(z))-\sum_{a_i\in\mathcal{A}_i}\overline{p}_{a_i}^t(z^t)\hat{r}^t(a_i^t)\bigg),\label{eq:general_infinite_Z_sleeping_regret}\\
        =&\sum_{t:z^t=z} \bigg(\hat{r}^t(\overline{\pi}_i^*(z))-\sum_{a_i\in\mathcal{A}_i}p_{a_i}^t(z^t)\hat{r}^t(a_i^t)\bigg),\label{eq:infinite_Z_conversion_sleeping_normal_regret}
    \end{align}

Equation \eqref{eq:general_infinite_Z_normal_regret} follows by definition of $\hat{\textbf{r}}$ (line $2$ Algorithm \ref{alg:c.z.no-regret_finite_Z}). Since we conditioned on equation \eqref{eq:simultaneous_UCB_LCB} in Lemma \ref{lem:union_confidence} holding true, $\mathbbm{1}_{\overline{\pi}_i^*(z),t}$ equals $1$ for all $t\geq 1$ and therefore equation \eqref{eq:general_infinite_Z_normal_regret} and equation \eqref{eq:general_infinite_Z_sleeping_regret} are equivalent. Equivalence between equation \eqref{eq:general_infinite_Z_sleeping_regret} and equation \eqref{eq:infinite_Z_conversion_sleeping_normal_regret} follows from the same arguments which are used in the proof of Theorem \ref{thm:general_algo_finite_Z} to show equivalence between equation \eqref{eq:general_finite_Z_sleeping_regret} and equation \eqref{eq:conversion_sleeping_normal_regret}. Note furthermore, that since we condition on equation \eqref{eq:simultaneous_UCB_LCB} in Lemma \ref{lem:union_confidence} holding true, $\text{UCB}_0^t(\cdot)\geq 0$ since $r(\cdot)\geq 0$ and thus $\hat{r}^t(\cdot)\in[0,1]$. Now observe that equation \eqref{eq:infinite_Z_conversion_sleeping_normal_regret} is precisely the regret with respect to expert $\overline{\pi}_i^*(z)$ which a player with reward function $\hat{r}^t(\cdot)\in[0, 1]$ incurs in an expert problem after $T_z=\sum_{t=1}^T 1_{\{\overline{z}^t=z\}}$ repetitions of the game. Here $T_z$ denotes the number of times the revealed context $z^t$ belonged to the ball centered at $z\in\mathcal{C}$. Recall that for each context $z\in\mathcal{C}$ a distribution $p^t(z)$ is maintained which is updated whenever context $z^t$ belongs to the ball with center $z$. The prediction $p^t(z^t)$ (Line $9$ in Algorithm \ref{alg:c.z.no-regret_infinite_Z}) is given by the no-regret algorithm $\mathcal{E}$ which receives full-information feedback $\hat{\textbf{r}}^t=[\hat{r}^t(a_i)]_{a_i\in\mathcal{A}_i}$. Thus, equation \eqref{eq:infinite_Z_conversion_sleeping_normal_regret} corresponds to the regret incurred by the algorithm $\mathcal{E}$ after $T_z$ repetitions and is upper-bounded by the regret bound of $\mathcal{E}$:
\begin{align}
    \sum_{t:\overline{z}^t=z} \bigg(\hat{r}^t(\overline{\pi}_i^*(z))-\sum_{a_i\in\mathcal{A}_i}p_{a_i}^t(z^t)\hat{r}^t(a_i^t)\bigg)\leq R^{T_z}(\mathcal{E}).
\end{align}
    
Summing over all the contexts $z\in\mathcal{C}$, we can further upper-bound equation \eqref{eq:infinite_Z_c.z.AGP} with probability at least $1-\delta/2$:
    \begin{align}
        R-R^T\leq&\sum_{z\in\mathcal{C}}\sum_{t:\overline{z}^t=z} \min\{1,\text{UCB}_0^t(\overline{\pi}_i^*(z),a_{-i}^t,z^t)\}- \sum_{a_i\in\mathcal{A}_i}\overline{p}_{a_i}^t(z^t)\min\{1,\text{UCB}_0^t(a_i,a_{-i}^t,z^t)\}\nonumber\\
        &+\sqrt{T/2\log(2/\delta)}+C_1\beta_0^T\sqrt{T\gamma_0^T}\nonumber\\
    &\leq \sum_{z\in\mathcal{C}}R^{T_z}(\mathcal{E})+\sqrt{T/2\log(2/\delta)}+C_1\beta_0^T\sqrt{T\gamma_0^T}\nonumber\\
    &\leq \sqrt{|\mathcal{C}|\sum_{z\in\mathcal{C}}(R^{T_z}(\mathcal{E}))^2}+\sqrt{T/2\log(2/\delta)}+C_1\beta_0^T\sqrt{T\gamma_0^T}\nonumber\\
    &\leq \epsilon^{-d/2}\sqrt{\sum_{z\in\mathcal{C}}(R^{T_z}(\mathcal{E}))^2}+\sqrt{T/2\log(2/\delta)}+C_1\beta_0^T\sqrt{T\gamma_0^T}.\label{eq:general_R-R_i^T}
\end{align}
In the last equality we used $\mathcal{C}\leq (1/\epsilon)^d$, because the contexts space $\mathcal{Z}\subseteq[0,1]^d$ can be covered by at most $(1/\epsilon)^d$ balls of radius $\epsilon$ such that the distance between their centers is at least $\epsilon$ \cite{clakson05}.

Then, with probability at least {\small$1-\frac{\delta}{2}$} the regret is bounded by:
\begin{align*}
    R^T&= L-R^T+R-R^T\\
    &\leq L_r L_p T\epsilon+\epsilon^{-d/2}\sqrt{\sum_{z\in\mathcal{C}}(R^{T_z}(\mathcal{E}))^2}+\sqrt{T/2\log(2/\delta)}+C_1\beta_0^T\sqrt{T\gamma_0^T}\\
    &=(L_r L_p)^{\frac{d}{d+2}} T^{\frac{d}{d+2}} \bigg(\sum_{z\in\mathcal{C}} (R^{T_z}(\mathcal{E}))^2\bigg)^{\frac{1}{d+2}}+\sqrt{T/2\log(2/\delta)}+C_1\beta_0^T\sqrt{T\gamma_0^T},
\end{align*}
where in the last equality $\epsilon$ is set as $(L_r L_p)^{-\frac{2}{d+2}}T^{-\frac{2}{d+2}}\bigg(\sum_{c\in\mathcal{C}}(R^{T_z}(\mathcal{E}))^2\bigg)^{\frac{1}{d+2}}$.
\end{proof}

\subsubsection{Proof of Proposition \ref{prop:c.z.CEE}}\label{app:game_equilibria}

In this section, we provide a proof for Proposition \ref{prop:c.z.CEE}.

For any player $i$ following a no-regret, no-violation algorithm the constrained contextual regret after $T$ rounds of game play is upper bounded by:
    \begin{align}\label{eq:proof_prop_1}
    R^T\geq&\max_{\substack{\pi_i\in\Pi_i}} \sum_{t=1}^T r(\pi_i(z^t),a_{-i}^t,z^t)-\sum_{t=1}^T r(a_i^t,a_{-i}^t,z^t)
\end{align}
for all $\pi_i\in\Pi_i$ such that $g_{i,m}(\pi_i(z^t))\leq 0$ for all $m\in[M]$ and $t\in[T]$.

By definition of the empirical joint policy $\rho^T$ at round $T$ (see equation \eqref{eq:emirical_joint_policy}) the following holds:
\begin{align}
\begin{split}\label{eq:proof_prop_2}
    \frac{1}{T} \sum_{t=1}^T \mathbb{E}_{a\sim\rho_{z^t}^T}\big[r_i(a_i,a_{-i},z^t)\big]&= \frac{1}{T} \sum_{t=1}^T \sum_{a\in\mathcal{A}} \rho_{z^t}^T(a)r_i(a_i,a_{-i},z^t)\\
    &\frac{1}{T} \sum_{t=1}^T \sum_{z\in\mathcal{Z}^T}\sum_{t: z^t=z} r_i(a_i^t,a_{-i}^t,z^t)\\
    &=\frac{1}{T} \sum_{t=1}^T  r_i(a_i^t,a_{-i}^t,z^t).
\end{split}
\end{align}
Following the same argumentation it further holds that:
\begin{align}\label{eq:proof_prop_3}
    \frac{1}{T} \sum_{t=1}^T\mathbb{E}_{a\sim\rho_{z^t}^T}\big[r_i(\pi_i(z^t),a_{-i},z^t)\big] = \frac{1}{T} \sum_{t=1}^T  r_i(\pi_i(z^t),a_{-i}^t,z^t).
\end{align}
By combining equation \eqref{eq:proof_prop_1}-\eqref{eq:proof_prop_3} and by definition of $\epsilon$ in Proposition \ref{prop:c.z.CEE} we obtain the following inequality:
    
\begin{align*}
     \frac{1}{T} \sum_{t=1}^T \mathbb{E}_{a\sim\rho_{z^t}^T}\big[r_i(a_i,a_{-i},z^t)\big]&\geq \frac{1}{T} \sum_{t=1}^T\mathbb{E}_{a\sim\rho_{z^t}^T}\big[r_i(\pi_i(z^t),a_{-i},z^t)\big] - \frac{R^T}{T}\\
     &\geq \frac{1}{T} \sum_{t=1}^T\mathbb{E}_{a\sim\rho_{z^t}^T}\big[r_i(\pi_i(z^t),a_{-i},z^t)\big] - \epsilon,
\end{align*}
for all $\pi_i\in\Pi_i$ such that $g_{i,m}(\pi_i(z^t))\leq 0$ for all $m\in[M]$ and $t\in[T]$.

To prove that $\rho^T$ is an $\epsilon$-c.z.CCE it is left to show that the time-averaged expected constraint violations are bounded by $\epsilon$.
For any player $i$ following a no-regret, no-violation algorithm the cumulative constraint violations after $T$ rounds of game play are upper bounded by:

\begin{align}\label{eq:proof_prop_5}
    \mathcal{V}_{i,m}\geq \sum_{t=1}^T [g_{i,m}(a_i^t,z^t)]_+
\end{align}
for all $m\in[M]$. Furthermore, it holds that:

\begin{align}
\begin{split}\label{eq:proof_prop_6}
    \frac{1}{T}\sum_{t=1}^T \mathbb{E}_{a\sim\rho_{z^t}^T}\big[[g_{i,m}(a_i,z^t)]_+\big]
    & = \frac{1}{T}\sum_{t=1}^T \rho_{z^t}^T(a)[g_{i,m}(a_i,z^t)]_+\\
    &= \frac{1}{T}\sum_{t=1}^T \sum_{z\in\mathcal{Z}^T}\sum_{t:z^t=z}[g_{i,m}(a_i^t,z^t)]_+\\
    &= \frac{1}{T}\sum_{t=1}^T [g_{i,m}(a_i^t,z^t)]_+
\end{split}
\end{align}

By combining equation \eqref{eq:proof_prop_5} and \eqref{eq:proof_prop_6} and by definition of $\epsilon$ in Proposition \ref{prop:c.z.CEE} we obtain the following inequality:
\begin{align*}
    \frac{1}{T}\sum_{t=1}^T\mathbb{E}_{a\sim\rho_{z^t}}\big[[g_{i,m}(a_i, z^t)]_+\big]&\leq \frac{\mathcal{V}_{i,m}^T}{T}\leq \epsilon, \quad\forall m\in[M]
\end{align*}
which completes the proof.

\subsection{Supplementary Material for Section \ref{sec:experiments}}\label{app:experiments}
In the following, we demonstrate c.z.AdaNormalGP on a random $N$-player game with unknown constraints. We furthermore, provide a complete description of our experimental setup for the application of controller design of robots presented in Section \ref{sec:experiments_robots}. 

\subsubsection{Random $N$-player games}\label{app:random_game}

We consider a repeated game between three players with action sets $\mathcal{A}_1=\mathcal{A}_2=\mathcal{A}_3=\{0,1,\ldots,K-1\}$ and contextual information of the form $\mathcal{Z}=\{1,\ldots,Z\}$. For each player, we simulate $10$ different reward functions $r_i:\mathcal{A}\times\mathcal{Z}\rightarrow[0,1]$ as the posterior mean computed from $10$ samples of a GP, i.e., $r_i\sim\mathcal{GP}(0,k_r(\cdot,\cdot))$. Kernel $k_r=k_a*k_z$ is a composite kernel which can encode different dependences of $r_i$ on $a\in\mathcal{A}$ and $z\in\mathcal{Z}$. Furthermore, the  actions of each player are subject to a constraint $g_i(a_i)\leq 0$ for all $a_i\in\mathcal{A}_i$. We simulate $10$ different constraint functions $g_i:\mathcal{A}_i\rightarrow\mathbb{R}$ as the posterior mean computed from $10$ samples of a GP, i.e., $g_i\sim\mathcal{GP}(0,k_g(\cdot,\cdot))$. Concretely, we set $K=7$, $Z=5$, and $T=1000$, the kernels are set as $k_a=k_{SE}$ with lengthscale $l=2$ and $k_z=k_g=k_{SE}$ with lengthscale $l=0.5$ and the noise is given by $\epsilon_i^t\sim\mathcal{N}(0,1)$. 

For every game, we consider player $1$ to either play 1) random, i.e., selects actions uniformly at random or according to 2) GPMW \cite{sessa2019no}, a no-regret algorithm, 3) z.GPMW \cite{sessa21}, a contextual no-regret algorithm, 4) c.AdaNormalGP, a version of c.z.AdaNormalGP which ignores contextual information, and 5) c.z.AdaNormalGP. Player $2$ and player $3$ are random players. The regret obtained by ``GPMW'', ``z.GPMW'', ``c.AdaNormalGP'' and ``c.z.AdaNormalGP'' decreases over $T$, shown in Figure \ref{fig:1} (\textit{left}). In particular, the two contextual no-regret algorithms ``z.GPMW'' and ``c.z.AdaNormalGP'' obtain smaller regret compared to their non-contextual counterparts since they exploit contextual information in their update rule. In addition, ``c.AdaNormalGP'' and ``c.z.AdaNormalGP'' learn the constraints on the action set and thus their cumulative constraint violations do not increase after a few rounds of game play, whereas ``GPMW'' and ``z.GPMW'' keep violating the constraints on the action set, shown in Figure \ref{fig:1} (\textit{right}).

\begin{figure}[t]
    \centering
    \includegraphics[scale = 0.57] {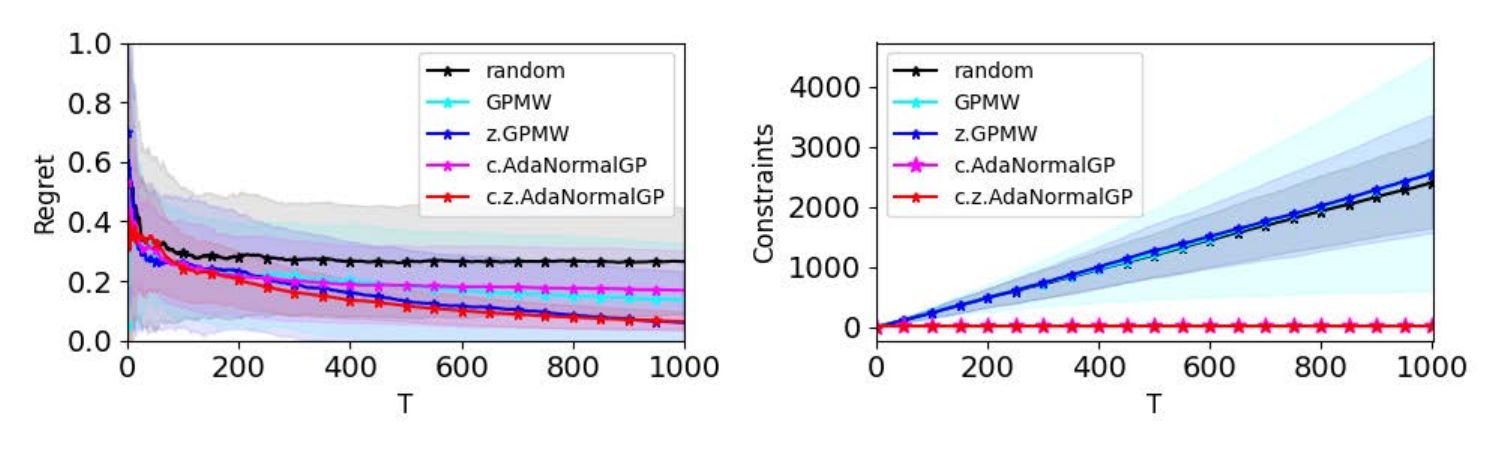}
    \caption{Regret and cumuluative constraint violations for players ``random'', ``GPMW'', ``c.GPMW'', ``c.AdaNormalGP'', and ``c.z.AdaNormalGP''. Shaded areas represent $\pm$ one standard deviation.} \label{fig:1}
\end{figure}

\end{document}